\documentclass[format=acmsmall, review=false]{acmart}

\pdfoutput=1

\usepackage{acm-ec-19}
\usepackage{booktabs}
\usepackage{microtype}
\usepackage{color}
\usepackage{url}
\usepackage{wrapfig}
\usepackage{enumitem}
\usepackage{tikz}
\usepackage{subcaption}

\newcommand{\phat}{\hat{p}}
\newcommand{\qhat}{\hat{q}}

\renewcommand{\epsilon}{\varepsilon}

\newcommand{\IM}{\textit{IM}}

\DeclareMathOperator{\med}{med}

\begin{document}
	\title{Truthful Aggregation of Budget Proposals}  
	
	\author{Rupert Freeman}
	\affiliation{Darden School of Business, University of Virginia}
	\author{David M. Pennock}
	\affiliation{DIMACS, Rutgers University}
	\author{Dominik Peters}
	\affiliation{School of Engineering and Applied Sciences, Harvard University}
	\author{Jennifer Wortman Vaughan}
	\affiliation{Microsoft Research New York City}

\begin{abstract} 
	We consider a participatory budgeting problem in which each voter submits a proposal for how to divide a single divisible resource (such as money or time) among several possible alternatives (such as public projects or activities) and these proposals must be aggregated into a single aggregate division. Under $\ell_1$ preferences---for which a voter's disutility is given by the $\ell_1$ distance between the aggregate division and the division he or she most prefers---the social welfare-maximizing mechanism, which minimizes the average $\ell_1$ distance between the outcome and each voter's proposal, is incentive compatible~\citep{GKSA16}. However, it fails to satisfy a natural fairness notion of \emph{proportionality}, placing too much weight on majority preferences. Leveraging a connection between market prices and the generalized median rules of \citet{Moul80}, we introduce the \emph{independent markets} mechanism, which is both incentive compatible and proportional. We unify the social welfare-maximizing mechanism and the independent markets mechanism by defining a broad class of \emph{moving phantom} mechanisms that includes both. We show that every moving phantom mechanism is incentive compatible. Finally, we characterize the social welfare-maximizing mechanism as the unique Pareto-optimal mechanism in this class, suggesting an inherent tradeoff between Pareto optimality and proportionality.
\end{abstract}

\maketitle

\section{Introduction}
\label{sec:introduction}

Participatory budgeting allows members of a community to collectively decide how to divide a portion of the community's budget among a set of proposed alternatives~\citep{Cab04}. Participatory budgeting was first introduced in Brazil~\citep{SG02}, and it has now been used in more than 3,000 cities around the world, including New York, Boston, Chicago, San Francisco, and Toronto~\citep{PBP}.  It has been used to determine how to allocate the budgets of states and cities as well as housing authorities and schools.

Many participatory budgeting elections are run using a variant of $k$-approval voting, in which each voter chooses up to $k$ projects to approve, and the projects with the highest number of approvals are funded, subject to budget constraints~\citep{GKSA16}.  Under such a voting scheme, each proposed project is either fully funded or not funded at all. This makes sense for well-delineated projects such as renovating a school or adding an elevator to a public library. For other kinds of projects, funding decisions need not be all-or-nothing.  For example, participatory budgeting could be used to decide how to divide a part of a city's spending between its departments of health, education, infrastructure, and parks. A voter might propose a division of the tax surplus among the four departments into the fractions (30\%, 40\%, 20\%, 10\%). The city could invite each citizen to submit such a budget proposal, and they could then be aggregated by a suitable mechanism.

A first idea for aggregating the proposals would be to take the arithmetic mean.
But the mean has a serious flaw as an aggregator: it's easily manipulated. A voter preferring a (60\%, 40\%) division across two alternatives may vote (100\%, 0\%) instead in order to distort the mean calculation and move the aggregate closer to his or her true preference if the first alternative has little community support.

In this paper, we seek mechanisms that are resistant to such manipulation.  In particular, we require that no voter can, by lying, move the aggregate division toward his or her preference on one alternative without moving it away from his or her preference by an equal or greater amount on other alternatives. In other words, we seek budget aggregation mechanisms that are incentive compatible under $\ell_1$ preferences, with each voter's disutility for a budget division equal to the $\ell_1$ distance between that division and the division she prefers most. 
One can show that such preferences emerge naturally from a model where voters have separable uniform utilities over alternatives together with a funding cap for each alternative beyond which it provides no value to the voter.

\citet{midpoint} and \citet{GKSA16} have shown that choosing an aggregate budget division that maximizes the utilitarian social welfare of the voters---that is, a division that minimizes the total $\ell_1$ distance from each voter's report---is both incentive compatible and Pareto optimal under this voter utility model.  However, this utilitarian aggregate has a tendency to overweight majority preferences. For example, imagine that one hundred voters prefer (100\%, 0\%) while ninety-nine prefer (0\%, 100\%). The utilitarian aggregate is (100\%, 0\%) even though the mean is close to (50\%, 50\%). In many participatory budgeting scenarios, the latter solution is more in the spirit of consensus. For example, suppose that we vote over the allocation of education funding to schools, and imagine that each family votes for all education funding to go to their own neighborhood school. The utilitarian aggregate would earmark the entire budget to the most populous school district. However, we may prefer that funds are split in proportion to the districts' populations.  To capture this fairness constraint, we define a notion of \emph{proportionality}, requiring that when voters are single-minded (as in this example), the fraction of the budget assigned to each alternative is equal to the proportion of voters who favor that alternative. Do there exist aggregators that are both incentive compatible and proportional?

This question is easiest to answer for two alternatives.
In this case, $\ell_1$ preferences are a special case of \emph{single-peaked} preferences, well studied in the voting literature. The seminal results of \citet{Moul80} imply that, in this setting, all incentive-compatible voting schemes correspond to inserting $n+1$ ``phantom'' proposals, where $n$ is the number of voters, and returning the median of the $n$ true proposals and the $n+1$ phantoms. We show that there exists a unique way of placing the phantoms that results in a proportional mechanism for two alternatives. 

Generalizing Moulin's phantom median mechanisms to allow for more than two alternatives is difficult. Existing proposals for such generalizations take a median in each dimension independently \citep{BJ83,PSS92,BM94}, but this strategy is doomed in our application with normalization constraints that require funding to sum up to 100\%. Unlike the mean, taking a coordinate-wise median will usually fail to normalize. We address this problem by allowing the set of phantoms to continuously shift upwards, increasing the sum of the generalized medians until the aggregate becomes normalized. This idea allows us to define a very general class of \emph{moving phantom} mechanisms. One might think that allowing the final phantom locations to depend on voters' reports might give voters an incentive to misreport. However, we prove that every moving phantom mechanism is incentive compatible under $\ell_1$ preferences.

Among this large family of incentive-compatible mechanisms, we find one that satisfies our proportionality requirement. This moving phantom mechanism is obtained when phantoms are placed uniformly between 0 and a value $x\ge 0$ which increases until the coordinate-wise medians sum to one. To analyze this mechanism, we prove that the aggregate found by this mechanism can be interpreted as the clearing prices in a market system, and hence call it the \emph{independent markets} mechanism. This reveals an unexpected connection between market prices and generalized medians that may be of broader interest. 
The independent markets mechanism can also be justified from a game-theoretic perspective as the unique Nash equilibrium of a natural voting game. Thus, this proportional moving phantom mechanism has two alternative interpretations:
\begin{enumerate}
\item Market interpretation: For each alternative, we set up a market in which $x$ units of a divisible good are sold. This amount $x$ is the same across all markets.  Each voter has a value for the good in market $j$ that is equal to the fraction of the budget that the voter would prefer be allocated to alternative $j$ in the budget division setting, and has \$1 to spend in each market.  Increase $x$ until the point at which the market clearing prices across these independent markets sum to \$1.  At this point, the market clearing prices are exactly the aggregate division selected by the independent markets moving phantom mechanism.

\item Voting game: Each agent receives one credit for each alternative and may choose any amount of that credit to spend on the alternative.  The outcome of the game is a normalized vector proportional to the amount spent on each alternative.  Agents choose their spending in such a way as to minimize the $\ell_1$ distance between their preferred budget division and the outcome of the game. This game has a unique Nash equilibrium outcome which is exactly the division selected by the independent markets mechanism.
\end{enumerate}
By analyzing the market and Nash equilibria of these systems, we can show that our mechanism satisfies several important social choice properties, such as participation and reinforcement. We also show that our mechanism retains its proportionality guarantees even for voters that approve exactly $k$ projects each, for some $k\ge 1$, and propose to split the budget equally among them.

In contrast, the independent markets mechanism fails to satisfy Pareto optimality.  We show that this is unavoidable, as no proportional moving phantom mechanism is Pareto optimal. 
In fact, we prove that there is a \textit{unique} moving phantom mechanism that is Pareto optimal (provided the number of alternatives is sufficiently large). In this mechanism, all phantoms start at 0 and then, one by one, transition to 1, with no two phantoms moving at the same time. 
This mechanism turns out to also have a phantom-free interpretation: it is equivalent to selecting the budget division that maximizes utilitarian social welfare (breaking ties in favor of divisions closer to a uniform division)---the same mechanism previously studied by \citet{midpoint} and \citet{GKSA16}.

While the motivation of our formal model is participatory budgeting, it applies to the division of other resources such as time.
For example, one might imagine using such a mechanism as a way to reach consensus among a team of conference organizers who wish to divide a day between talks, poster sessions, and social activities.
Another example would be a government that needs to decide on a target energy mix (that is, how much energy should come from fossil fuels, nuclear, or renewable sources) and wishes to aggregate expert proposals.
In all these applications, our class of moving phantom mechanisms can be used to make better decisions.

\paragraph{Related Work}

Several recent papers study voting rules for participatory budgeting, considering both axiomatics and computational complexity, but under the assumption that projects are indivisible, so can either be fully funded or not funded at all~\citep{aziz2020participatory,GKSA16,BPNS17,LB11c,ALT18}.  
The setting in which partial funding of alternatives is permitted has also been studied, but generally under a different utility model in which voters assign utility scores to the alternatives rather than having an ideal distribution~\citep{FGM16,BMS05,ABM17,airiau2019portioning}. In this model, \citet{GuNe14a} characterize a rule based on maximizing the Nash product using an axiom very similar to our proportionality property: If voters are single-minded in the sense of assigning utility 0 to all but one project, then the projects receive funding proportional to the number of supporters. \citet{Dudd15a} proposes a rule for this setting that satisfies a stronger proportionality criterion (which gives representation guarantees to every subset of voters) and is strategyproof if voters have dichotomous preferences. This body of work includes positive results for weak versions of incentive compatibility, but impossibilities for obtaining full incentive compatibility.
\citet{GKG+18} perform a Mechanical Turk study exploring preference structure in a high-dimensional continuous setting similar to ours.

Closely related to our work is an unpublished paper by \citet{midpoint}, as well as the PhD thesis (in German language) of \citet{lindner2011manipulierbarkeit}. Like us, they study the aggregation of points in the standard simplex. Their work focusses on the set of points that maximize utilitarian social welfare, defined in the $\ell_1$ sense, and studies its structure. Further, they discuss a sense in which this set, considered as an irresolute voting rule, is strategyproof  \citep[Sec.~3.2.2]{lindner2011manipulierbarkeit}. \citet[Sec.~3.2.4]{lindner2011manipulierbarkeit} introduces a tie-breaking scheme to turn this rule into a resolute (single-valued) aggregation rule, and gives several interpretations of this scheme. The same tie-breaking scheme also emerges in our treatment as the unique neutral way to break ties compatible with moving phantom mechanisms. \citet[Sec.~3.4]{lindner2011manipulierbarkeit} then proposes several other aggregation rules inspired by different notions of multi-dimensional medians, some of which turn out to be strategyproof. Finally, \citet[Sec.~4]{lindner2011manipulierbarkeit} presents results of a detailed simulation study.

Also closely related is the work of \citet{GKSA16}. The primary focus of their paper is on \emph{knapsack voting}, in which each voter submits her preferred set of projects to fully fund.  However, they also consider the use of \emph{fractional} knapsack voting in a setting in which partial funding of alternatives is permitted and voters have $\ell_1$ preferences, which is exactly our setting. They re-discover the result of \citet{midpoint} that the mechanism that maximizes social welfare is incentive compatible, though they use a different tie-breaking rule (namely, lexicographic). \citeauthor{GKSA16} do not consider other mechanisms for this setting.

More broadly, the truthful aggregation of preferences over numerical values (such as the temperature for an office) has been extensively studied. A famous result of \citet{Moul80} characterizes the set of incentive-compatible voting rules under the assumption that voters have single-peaked preferences over values in $[0,1]$. These voting rules are generalized median schemes. The best-known example is the standard median, in which each voter reports her ideal point in $[0,1]$ and the median report is selected. Other voting rules in this class insert ``phantom voters'' who report a fixed top choice. \citet{Barbera93:Generalized} obtained a multi-dimensional analogue of this result for $[0,1]^m$, and there are further generalizations that characterize truthful rules if other constraints are imposed on the feasible set \citep{Barbera97:Voting}. Crucially, the constraints allowed by \citet{Barbera97:Voting} do not include the normalization constraint 
that is fundamental to our setting. 
Several other papers [\citealp{BJ83,PSS92,BM94}; see \citealp[Section 5]{bordes2011euclidean}, for an overview] introduced multidimensional models in which one can achieve truthfulness by taking a generalized median in each coordinate, but such a strategy does not work with normalization constraints. 
In many of these models, all strategyproof mechanisms decompose into one-dimensional mechanisms [e.g., \citealp[Theorem 1]{BJ83}].
Prior to this work, we are only aware of one multidimensional model in which truthful mechanisms that do not decompose were found: the aggregation of points in $\mathbb R^d$, where $n = 2$ agents report an ideal point and have $\ell_2$ preferences [\citealp{laffond1980revelation}; \citealp[Example 2]{BJ83}; \citealp[Example 39]{kim1984nonmanipulability}], though none of these mechanisms can be Pareto optimal \citep[Theorem 41]{kim1984nonmanipulability}.

In the computer science literature, the above-mentioned generalized median schemes have also been studied in the context of truthful facility location \citep{PT13}. In this context, the aim is to approximate the optimum social welfare subject to incentive compatibility.

One could apply our results to the aggregation of probabilistic beliefs. There is a large literature on \emph{probabilistic opinion pooling} \citep{Genest86,French85,Clemen89,intriligator1973probabilistic} which studies aggregators in this context. The main focus of that literature is to preserve stochastic and epistemic properties. To the best of our knowledge, strategic aspects have not been considered.

Finally, recently proposed rules for crowdsourcing societal tradeoffs
\citep{Conitzer15:Crowdsourcing,conitzer2016rules} can also be used to aggregate budget divisions (with full support) after converting them into pairwise ratios of funding amounts, but this setting has also not been analyzed from a strategic viewpoint.

 \section{Preliminaries}
\label{sec:preliminaries}

For $k \in \mathbb{N}$, we write $[k] = \{ 1, \ldots, k \}$. Let $N = [n]$ be a set of voters and $M = [m]$ be a set of possible alternatives. Voters have structured preferences over budget \emph{divisions} $\mathbf{p} \in [0,1]^m$, with $\sum_{j \in [m]} p_j=1$, where $p_j$ is the fraction of a public resource (such as money or time) allocated to alternative $j$. Each voter $i$ has a most preferred division (or \emph{ideal point}) $\mathbf{p}_i=(p_{i,1},\ldots,p_{i,m})$, with their preference over other divisions induced by $\ell_1$ distance from $\mathbf{p}_i$. 
Specifically, each voter $i$ has a disutility for division $\mathbf{q}$ equal to $d(\mathbf{p}_i,\mathbf{q})$, where $d(\mathbf{x},\mathbf{y}) = \sum_{j=1}^m |x_j-y_j|$
denotes the $\ell_1$ distance between $\mathbf{x}$ and $\mathbf{y}$.
Note that a voter's complete preference over all possible divisions can be deduced from their most preferred division $\mathbf{p}_i$.

\citet[Sec.~3.2.2]{nehringpuppe} study $\ell_1$ preferences in a related setting, and suggest the following model where they arise naturally:
Suppose voter $i$ has an additively separable utility function $u_\mathbf{x}(\mathbf{q}) = \sum_{j=1}^m \min (q_j, x_j)$, where $x_1, \dots, x_m \in [0,1]$ are \emph{funding caps} that sum to 1. According to this specification, for every alternative $j$, voter $i$ derives utility linear in the amount spent on $j$, with no additional utility for spending over and above the cap $x_j$.\footnote{A more expressive class of utilities would allow agents to obtain utility from different alternatives at different rates, up to the cap. \citet{nehringpuppe} argue that this restricted model has the advantage of having low informational requirements, and that the assumption of equal rates can be justified from the principle of insufficient reason.}
Note that the distribution maximizing $u_\mathbf{x}$ is $\mathbf{q} = (x_1, \dots, x_m)$. In fact, when the funding caps are determined by voter $i$'s preference $\mathbf{p}_i$, we can check that $i$ has $\ell_1$ preference with ideal point $\mathbf p_i$: note that $u_{\mathbf{p}_i}(\mathbf{q}) = 1- \sum_{j \in \{ h \in [m] : q_h > p_{i,h} \} } (p_{i,j}-q_j)=1-\frac{1}{2} d(\mathbf{p}_i,\mathbf{q})$, and hence voter $i$ prefers one distribution to another according to $u_{\mathbf p_i}$ if and only if it is $\ell_1$-closer to $\mathbf p_i$.

A \emph{preference profile} $\mathbf{P} = ( \mathbf{\phat}_1, \mathbf{\phat}_2, \ldots, \mathbf{\phat}_n )$ consists of a reported division $\mathbf{\phat}_i$ for each voter $i$. We use $\mathbf{P}_{-i}$ to denote the reports of all voters other than $i$. A \emph{budget aggregation mechanism} $\mathcal{A}$ takes as input a preference profile $\mathbf{P}$, and outputs an aggregate division $\mathcal{A}(\mathbf{P})$. 
A mechanism is \emph{continuous} if it is continuous when considered as a function $\mathcal{A}: (\mathbb{R}^m)^n \to \mathbb{R}^m$.
We say that a mechanism is \emph{anonymous} if its output is fixed under permutations of the voters, and \emph{neutral} if a permutation of the alternatives in voters' inputs permutes the output in the same way.

We are interested in mechanisms that satisfy \emph{incentive compatibility}, in the sense of dominant strategy incentive compatibility (also known as strategyproofness). Voters should not be able to change the aggregate division in their favor by misrepresenting their preference. 

\begin{definition}
	\label{def:ic}
	A budget aggregation mechanism $\mathcal{A}$ satisfies incentive compatibility if, for all preference profiles $\mathbf{P}$, voters $i$,
and divisions $\mathbf{p}_i$ and $\mathbf{\phat}_i$, 
$d(\mathcal{A}(\mathbf{\phat}_i,\mathbf{P}_{-i}),\mathbf{p}_i) \ge d(\mathcal{A}(\mathbf{p}_i,\mathbf{P}_{-i}),\mathbf{p}_i)$.
\end{definition}

We are also interested in the basic efficiency notion of \emph{Pareto optimality}. It should not be possible to change the aggregate so that some voter is strictly better off but no other voter is worse off.

\begin{definition}
	A budget aggregation mechanism $\mathcal{A}$ satisfies Pareto optimality if, for all preference profiles $\mathbf{P}$, and all divisions $\mathbf{q}$, if
$d(\mathcal{A}(\mathbf{P}),\mathbf{\phat}_i) > d(\mathbf{q},\mathbf{\phat}_i)$
for some voter $i$, then there exists a voter $j$ for which
$d(\mathcal{A}(\mathbf{P}),\mathbf{\phat}_j) < d(\mathbf{q},\mathbf{\phat}_j)$.
\label{def:po}
\end{definition}

Observe that the definitions of incentive compatibility and Pareto optimality depend only on the voters' preference relations, not the exact utility model. Results pertaining to these properties therefore hold for any utility function that induces the same ordinal preferences as $\ell_1$ utilities.

We also consider a fairness property that we call proportionality:
Suppose each voter is single-minded, in that their ideal point is a division in which the entire budget goes to a single alternative. Then it is natural to split the resource in proportion to the number of voters supporting each alternative. For example, if 6 voters report $(1,0,0)$, 3 voters report $(0,1,0)$, and 1 voter reports $(0,0,1)$, then the aggregate should be $(0.6,0.3,0.1)$. We call this property \emph{proportionality}.

\begin{definition}
	\label{def:proportional}
	A voter is \emph{single-minded} if their preferred division is a unit vector.
	A budget aggregation mechanism $\mathcal{A}$ is \emph{proportional} if, for every preference profile $\mathbf{P}$ consisting of only single-minded voters, and every alternative $j$, we have $\mathcal{A}(\mathbf{P})_j = \sum_{j \in [m]} \phat_{i,j}/n$
\end{definition}

We note that proportionality is a fairly weak axiom, only applying to a small subset of possible profiles. However, as we will see later, it is already strong enough to be incompatible with Pareto optimality within the class of moving phantom mechanisms that we introduce in this paper.

Finally, we introduce two properties from social choice theory that make sense for mechanisms that are defined for instances with a variable number of voters.
One of these is participation, which says that voters should always prefer (truthfully) participating to not contributing a report.

\begin{definition}
	A budget aggregation mechanism $\mathcal{A}$ satisfies \emph{participation} if, for all $\mathbf{p}_i$ and all preference profiles $\mathbf{P}_{-i}$, $d(\mathcal{A}(\mathbf{p}_i,\mathbf{P}_{-i}),\mathbf{p}_i) \le d(\mathcal{A}(\mathbf{P}_{-i}),\mathbf{p}_i)$.
\end{definition}

The second property is reinforcement, which states that if we have two preference profiles which yield the same aggregate division, then the profile obtained by combining these profiles should yield that aggregate division as well. For $\mathbf{P} = ( \mathbf{\phat}_1, \mathbf{\phat}_2, \ldots, \mathbf{\phat}_{n_P} )$ and $\mathbf{Q} = ( \mathbf{\qhat}_1, \mathbf{\qhat}_2, \ldots, \mathbf{\qhat}_{n_Q} )$, let $\mathbf{P} \cup \mathbf{Q} = ( \mathbf{\phat}_1, \ldots, \mathbf{\phat}_{n_P}, \mathbf{\qhat}_1, \ldots, \mathbf{\qhat}_{n_Q} )$.

\begin{definition}
	A budget aggregation mechanism $\mathcal{A}$ satisfies reinforcement if, for all preference profiles $\mathbf{P}$ and $\mathbf{Q}$ with $\mathcal{A}(\mathbf{P})=\mathcal{A}(\mathbf{Q})$, it holds that $\mathcal{A}(\mathbf{P}\cup \mathbf{Q}) = \mathcal{A}(\mathbf{P})=\mathcal{A}(\mathbf{Q})$.
\end{definition} \section{Two alternatives}
\label{sec:binary}

To build intuition, we begin by considering the case in which $m=2$. 
Due to the normalization of inputs and of the output, and with $\ell_1$ preferences, the problem is one-dimensional: a budget division is completely determined by the fraction spent on the first of the two alternatives.
This allows us to directly import Moulin's \citeyearpar{Moul80} famous characterization of \emph{generalized median} rules as the only incentive compatible mechanisms for voters with single-peaked preferences over a single-dimensional quantity.\footnote{Our preference model using $\ell_1$ imposes slightly more structure than just single-peakedness, namely that voters are indifferent between points that are equidistant to their peak. However, this restriction does not enlarge the class of incentive compatible mechanisms, at least if we impose continuity \citep{MB11}.}
\begin{theorem}[\citealp{Moul80}]
	\label{thm:binary}
	For $m=2$, an anonymous and continuous budget aggregation mechanism $\mathcal{A}$ is incentive compatible if and only if there are $\alpha_0 \ge \alpha_1 \ge \ldots \ge \alpha_{n}$ in $[0,1]$ such that, for all profiles $\mathbf P$, \begin{align*}
		\mathcal{A}(\mathbf P)_1 &= \med(p_{1,1}, p_{2,1}, \ldots, p_{n,1}, \alpha_0, \alpha_1, \ldots, \alpha_{n}), \\
		\mathcal{A}(\mathbf P)_2 &= \med(p_{1,2}, p_{2,2}, \ldots, p_{n,2}, 1-\alpha_0, 1-\alpha_1, \ldots, 1-\alpha_{n}).
	\end{align*}
\end{theorem}
The numbers $\alpha_k$ are known as \emph{phantoms}. Each mechanism described by Theorem~\ref{thm:binary} can be understood as taking the coordinate-wise median of the reported distributions, after inserting $n+1$ phantom voters (whose report is fixed and independent of the input profile). 

One can check that $\alpha_0, \dots, \alpha_{n}$ define a neutral mechanism if and only if the phantom placements are symmetric, that is if and only if
$ (\alpha_0, \dots, \alpha_{n} ) = (1-\alpha_n, \dots, 1-\alpha_{0} ). $
Note that there are $n+1$ phantoms but only $n$ voters, so that the phantoms can outweigh the voters. For example, when $\alpha_k = 1/2$ for all $k \in \{0,\dots,n\}$ then the mechanism is just the constant mechanism returning $(1/2,1/2)$.
However, if we take $\alpha_0 = 1$ and $\alpha_n = 0$, then these two phantoms ``cancel out'' and there are only $n-1$ phantoms left.
In fact, one can check that the mechanism is Pareto optimal if and only if  $\alpha_0=1$ and $\alpha_{n}=0$~\citep{Moul80}.

A particularly interesting example is the \emph{uniform phantom mechanism}, obtained when placing the phantoms uniformly over the interval $[0,1]$, so that $\alpha_k = 1-k/n$ for each $k \in \{0,\dots,n\}$. This placement of phantom voters appears in a paper by \citet{CPS16b}. They were aiming for mechanisms whose outcome is close to the mean, and they prove that the uniform phantom mechanism yields an aggregate that is closer to the mean than that obtained from any other phantom placements, in the worst case over all profiles.
The uniform phantom mechanism has other attractive properties, including being proportional in the sense of Definition~\ref{def:proportional}.

\begin{proposition} 
	\label{prop:uniform-proportional}
	For $m=2$, the uniform phantom mechanism is the unique (anonymous and continuous) budget aggregation mechanism $\mathcal A$ that is both incentive compatible and proportional.
\end{proposition}

\begin{proof}
  Theorem~\ref{thm:binary} gives us that $\mathcal{A}$ is incentive compatible if and only if it can be written in terms of phantom medians. We therefore need only to consider the additional requirement of proportionality.
	The uniform phantom mechanism is proportional, because if $\mathbf P$ consists of $n-k$ voters reporting $(1,0)$ and $k$ voters reporting $(0,1)$, then $\mathbf{A}(\mathbf P)_1 = \alpha_k = (n-k)/n$, as required.
	
	For uniqueness, suppose $\alpha_0, \dots, \alpha_n$ are phantom positions that induce a proportional mechanism. 
	Let $k\in \{ 0, \ldots, n \}$. We show that $\alpha_k = 1-k/n$. 
	Let $\mathbf{P}$ be a profile consisting of only single-minded voters with $n_1 = n-k$ voters reporting $\mathbf{\phat}_i=(1,0)$. Then $\alpha_k$ is the median, and proportionality requires that $\alpha_k = n_1/n = (n-k)/n = 1 - k/n$.
\end{proof}

Another natural way to place the phantoms is one that takes the coordinate-wise median.
When $n+1$ is even, this is achieved by placing half the phantoms at 0 and the other half at 1, outputting precisely the median of the reported values on each coordinate. When $n+1$ is odd, we place $n/2$ phantoms at 0, $n/2$ phantoms at 1, and we place a single phantom at $1/2$ to preserve neutrality. This mechanism outputs the point between the left and right medians that is closest to $1/2$. 
The resulting mechanism returns an aggregate $\mathbf p$ that minimizes the sum of distances between the reports $(\mathbf{p}_1, \mathbf{p}_2, \ldots, \mathbf{p}_n)$ and $\mathbf p$.
We will generalize this mechanism to larger $m$ in Section~\ref{sec:utilitarian}. 

\section{A Class of Incentive Compatible Mechanisms for Higher Dimensions}
\label{sec:higher-dimension}

For $m = 2$, we have a complete picture of incentive-compatible mechanisms, thanks to Moulin's characterization. For $m \ge 3$, it is less clear how to construct examples of incentive-compatible mechanisms. One could try to take a generalized median in each alternative independently, but the result of such a mechanism would not respect the normalization constraint.

\begin{example}
	\label{ex:uniform-not-normalized}
	Let us consider a simple numerical example. Let $n=m=3$, and suppose voter reports are $\mathbf{p}_1=(0, 0.5, 0.5), \mathbf{p}_2=(\frac{1}{3}, \frac{2}{3}, 0)$, and $\mathbf{p}_3=(0.9,0,0.1)$. Using the uniform phantom mechanism to aggregate the votes on each alternative independently yields $(\frac{1}{3}, 0.5, \frac{1}{3})$. Note that these sum to more than 1 (indeed, as we will see later, this is true in general for the uniform phantom mechanism on more than two alternatives). 
	
	A simple modification is to normalize the output of the uniform phantom mechanism by the appropriate multiplicative constant. In this example, we need to divide by $\frac{7}{6}$, which yields $\mathbf{q} = (\frac{2}{7}, \frac{3}{7}, \frac{2}{7})$. Unfortunately, this mechanism is not incentive compatible. To see this, suppose that the second voter instead reports $\mathbf{\phat}_2 = (\frac{3}{8}, \frac{5}{8}, 0)$. The normalized uniform phantom mechanism now selects the distribution $\mathbf{q}' = (\frac{9}{29}, \frac{12}{29}, \frac{8}{29})$. But $d(\mathbf{q}, \mathbf{p}_2) = \frac{4}{7}$ and $d(\mathbf{q}', \mathbf{p}_2) = \frac{16}{29} < \frac{4}{7}$, so voter 2 obtained a preferred outcome by manipulating.
\end{example}

However, there is a way of extending the idea of generalized medians to the higher-dimensional setting. The basic idea is that if a coordinate-wise generalized median violates the normalization constraint, then we can adjust the placement of the phantoms, increasing or decreasing the sum of the generalized medians as needed. Such a procedure might, in principle, give voters incentives to manipulate in order to affect the phantom placements. However, our class of \emph{moving phantom mechanisms} manages to avoid this problem.

\begin{definition}
	\label{def:moving-phantom}
	Let $\mathcal{F} = \{ f_k : k \in \{0, \ldots, n \} \}$ be a family of functions, which we call a  \emph{phantom system}, where $f_k : [0, 1] \to [0,1]$ is a continuous, weakly increasing function with $f_k(0)=0$ and $f_k(1)=1$ for each $k$, and we have $f_0(t) \ge f_1(t) \ge \cdots \ge f_n(t)$ for all $t\in[0,1]$.
	Then, the \emph{moving phantom mechanism} $\mathcal{A}^\mathcal{F}$ is defined so that for all profiles $\mathbf P$ and all $j \in [m]$,
	\begin{equation}
	\label{eq:Af-def}
	\mathcal{A}^\mathcal{F}(\mathbf{P})_j = \med(f_0(t^*), \ldots, f_{n}(t^*), \phat_{1,j}, \ldots, \phat_{n,j}),
	\end{equation}
	where $t^*$ is chosen so that $t^* \in \{ t : \sum_{j \in [m]} \med(f_0(t), \ldots, f_{n}(t),\phat_{1,j}, \ldots, \phat_{n,j})=1 \}$. 
	For brevity, we write $\mathcal{F}(t) = (f_0(t), \ldots, f_{n}(t))$ and abbreviate the median in \eqref{eq:Af-def} to $\med(\mathcal{F}(t),\mathbf{P}_{i \in [n],j})$.
\end{definition}

Let us examine the definition. Each $f_k$ represents a phantom, and the phantom system $\mathcal{F}$ represents a ``movie'' in which all phantoms continuously increase from 0 to 1, with the function argument $t$ defining an instantaneous snapshot of the phantom positions. The moving phantom mechanism $\mathcal{A}^\mathcal{F}$ defined by $\mathcal{F}$ identifies a particular snapshot in time, $t^*$, for which the sum of generalized medians over all coordinates is exactly 1. One can check that at least one such $t^*$ exists, and that the output of the mechanism is independent of which of these $t^*$ is chosen.

\begin{proposition}
	\label{prop:existence}
	The moving phantom mechanism $\mathcal{A}^\mathcal{F}$ is well-defined for every phantom system $\mathcal{F}$ satisfying the conditions of Definition~\ref{def:moving-phantom}.
\end{proposition}

\begin{proof}
	First note that the function $t \mapsto \sum_{j \in [m]} \med(\mathcal{F}(t),\mathbf{P}_{i \in [n],j})$ is continuous and increasing in $t$, because $f_k$ is continuous and increasing, and these properties are preserved under taking the median and sum.
	This implies that, provided the set $\{ t : \sum_{j \in [m]} \med(\mathcal{F}(t),\mathbf{P}_{i \in [n],j})=1 \}$ is non-empty, the aggregate $\mathcal{A}^\mathcal{F}(\mathbf{P})$ does not depend on the choice of $t^*$.

	When $t=0$, $\sum_{j \in [m]} \med(\mathcal{F}(t),\mathbf{P}_{i \in [n],j})= 0$, since all $n+1$ phantom entries are 0.
	When $t=1$, $\sum_{j \in [m]} \med(\mathcal{F}(t),\mathbf{P}_{i \in [n],j})= m > 1$, since all $n+1$ phantom entries are 1.
	By the Intermediate Value Theorem, using continuity, there exists $t \in [0,1]$ with $\sum_{j \in [m]} \med(\mathcal{F}(t),\mathbf{P}_{i \in [n],j})=1$.
\end{proof}

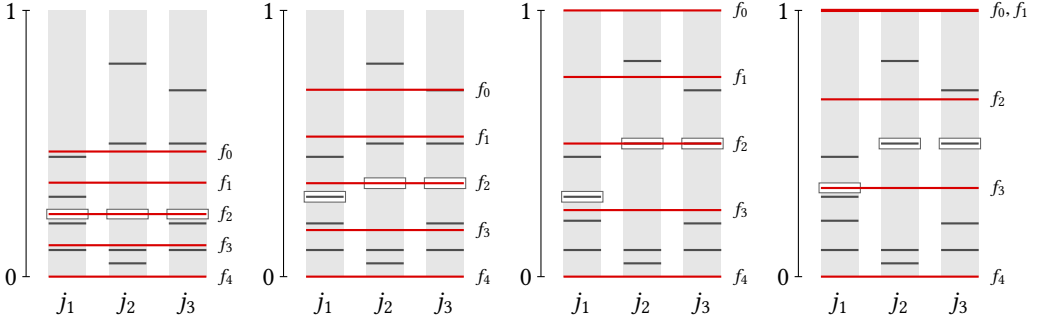
\begin{figure*}[!t]
	\centering
	\begin{subfigure}[b]{0.24\textwidth}
		\begin{tikzpicture}
		[yscale=3.5, vote/.style={black!70,thick}]
\draw (0,0) -- (0,1); \draw (-0.05,1) -- (0.05,1); \draw (-0.05,0) -- (0.05,0); \node (y0) at (-0.2,0) {0};
		\node (y1) at (-0.2,1) {1};
		
\begin{scope}[xshift=-0.2cm]
		\fill[fill=black!10] (0.5,0) rectangle (1,1);
\draw[black!60, thin, fill=white] (0.475,0.217) rectangle (1.025,0.253);
		\node (x1) at (0.75,-0.1) {$j_1$};
		\draw[vote] (0.5,0.3) -- (1,0.3);
		\draw[vote] (0.5,0.45) -- (1,0.45);
		\draw[vote] (0.5,0.2) -- (1,0.2);
		\draw[vote] (0.5,0.1) -- (1,0.1);
\end{scope}
		
		\begin{scope}[xshift=0.6cm]
\fill[fill=black!10] (0.5,0) rectangle (1,1);
\draw[black!60, thin, fill=white] (0.475,0.217) rectangle (1.025,0.253);
		\node (x1) at (0.75,-0.1) {$j_2$};
		\draw[vote] (0.5,0.5) -- (1,0.5);
		\draw[vote] (0.5,0.05) -- (1,0.05);
		\draw[vote] (0.5,0.1) -- (1,0.1);
		\draw[vote] (0.5,0.8) -- (1,0.8);
\end{scope}
		
		\begin{scope}[xshift=1.4cm]
\fill[fill=black!10] (0.5,0) rectangle (1,1);
\draw[black!60, thin, fill=white] (0.475,0.217) rectangle (1.025,0.253);
		\node (x1) at (0.75,-0.1) {$j_3$};
		\draw[vote] (0.5,0.2) -- (1,0.2);
		\draw[vote] (0.5,0.5) -- (1,0.5);
		\draw[vote] (0.5,0.7) -- (1,0.7);
		\draw[vote] (0.5,0.1) -- (1,0.1);
\end{scope}

\draw[red!85!black,thick] (0.3,0.0) -- (2.4,0.0);
		\draw[red!85!black,thick] (0.3,0.118) -- (2.4,0.118);
		\draw[red!85!black,thick] (0.3,0.235) -- (2.4,0.235);
		\draw[red!85!black,thick] (0.3,0.353) -- (2.4,0.353);
		\draw[red!85!black,thick] (0.3,0.470) -- (2.4,0.470);
		
		\node[text width=1cm, font=\tiny, anchor=west] at (2.45,0.470) {$f_0$};		
		\node[text width=1cm, font=\tiny, anchor=west] at (2.45,0.353) {$f_1$};		
		\node[text width=1cm, font=\tiny, anchor=west] at (2.45,0.235) {$f_2$};		
		\node[text width=1cm, font=\tiny, anchor=west] at (2.45,0.118) {$f_3$};		
		\node[text width=1cm, font=\tiny, anchor=west] at (2.45,0.0) {$f_4$};		
\end{tikzpicture}
\end{subfigure}
	\begin{subfigure}[b]{0.24\textwidth}
		\begin{tikzpicture}
		[yscale=3.5, vote/.style={black!70,thick}]
\draw (0,0) -- (0,1); \draw (-0.05,1) -- (0.05,1); \draw (-0.05,0) -- (0.05,0); \node (y0) at (-0.2,0) {0};
		\node (y1) at (-0.2,1) {1};
		
\begin{scope}[xshift=-0.2cm]
		\fill[fill=black!10] (0.5,0) rectangle (1,1);
		\draw[black!60, thin, fill=white] (0.475,0.28) rectangle (1.025,0.32);
		\node (x1) at (0.75,-0.1) {$j_1$};
		\draw[vote] (0.5,0.3) -- (1,0.3);
		\draw[vote] (0.5,0.45) -- (1,0.45);
		\draw[vote] (0.5,0.2) -- (1,0.2);
		\draw[vote] (0.5,0.1) -- (1,0.1);
\end{scope}
		
		\begin{scope}[xshift=0.6cm]
\fill[fill=black!10] (0.5,0) rectangle (1,1);
		\draw[black!60, thin, fill=white] (0.475,0.331) rectangle (1.025,0.371);
		\node (x1) at (0.75,-0.1) {$j_2$};
		\draw[vote] (0.5,0.5) -- (1,0.5);
		\draw[vote] (0.5,0.05) -- (1,0.05);
		\draw[vote] (0.5,0.1) -- (1,0.1);
		\draw[vote] (0.5,0.8) -- (1,0.8);
\end{scope}
		
		\begin{scope}[xshift=1.4cm]
\fill[fill=black!10] (0.5,0) rectangle (1,1);
		\draw[black!60, thin, fill=white] (0.475,0.331) rectangle (1.025,0.371);
		\node (x1) at (0.75,-0.1) {$j_3$};
		\draw[vote] (0.5,0.2) -- (1,0.2);
		\draw[vote] (0.5,0.5) -- (1,0.5);
		\draw[vote] (0.5,0.7) -- (1,0.7);
		\draw[vote] (0.5,0.1) -- (1,0.1);
\end{scope}
		
\draw[red!85!black,thick] (0.3,0.0) -- (2.4,0.0);
		\draw[red!85!black,thick] (0.3,0.175) -- (2.4,0.175);
		\draw[red!85!black,thick] (0.3,0.351) -- (2.4,0.351);
		\draw[red!85!black,thick] (0.3,0.526) -- (2.4,0.526);
		\draw[red!85!black,thick] (0.3,0.702) -- (2.4,0.702);
		
		\node[text width=1cm, font=\tiny, anchor=west] at (2.45,0.702) {$f_0$};		
		\node[text width=1cm, font=\tiny, anchor=west] at (2.45,0.526) {$f_1$};		
		\node[text width=1cm, font=\tiny, anchor=west] at (2.45,0.351) {$f_2$};		
		\node[text width=1cm, font=\tiny, anchor=west] at (2.45,0.175) {$f_3$};		
		\node[text width=1cm, font=\tiny, anchor=west] at (2.45,0.0) {$f_4$};		
\end{tikzpicture}
\end{subfigure}
	\begin{subfigure}[b]{0.24\textwidth}
		\begin{tikzpicture}
		[yscale=3.5, vote/.style={black!70,thick}]
\draw (0,0) -- (0,1); \draw (-0.05,1) -- (0.05,1); \draw (-0.05,0) -- (0.05,0); \node (y0) at (-0.2,0) {0};
		\node (y1) at (-0.2,1) {1};
		
\begin{scope}[xshift=-0.2cm]
		\fill[fill=black!10] (0.5,0) rectangle (1,1);
		\draw[black!60, thin, fill=white] (0.475,0.28) rectangle (1.025,0.32);
		\node (x1) at (0.75,-0.1) {$j_1$};
		\draw[vote] (0.5,0.3) -- (1,0.3);
		\draw[vote] (0.5,0.45) -- (1,0.45);
		\draw[vote] (0.5,0.21) -- (1,0.21);
		\draw[vote] (0.5,0.1) -- (1,0.1);
\end{scope}
		
		\begin{scope}[xshift=0.6cm]
\fill[fill=black!10] (0.5,0) rectangle (1,1);
		\draw[black!60, thin, fill=white] (0.475,0.48) rectangle (1.025,0.52);
		\node (x1) at (0.75,-0.1) {$j_2$};
		\draw[vote] (0.5,0.5) -- (1,0.5);
		\draw[vote] (0.5,0.05) -- (1,0.05);
		\draw[vote] (0.5,0.1) -- (1,0.1);
		\draw[vote] (0.5,0.81) -- (1,0.81);
\end{scope}
		
		\begin{scope}[xshift=1.4cm]
\fill[fill=black!10] (0.5,0) rectangle (1,1);
		\draw[black!60, thin, fill=white] (0.475,0.48) rectangle (1.025,0.52);
		\node (x1) at (0.75,-0.1) {$j_3$};
		\draw[vote] (0.5,0.2) -- (1,0.2);
		\draw[vote] (0.5,0.5) -- (1,0.5);
		\draw[vote] (0.5,0.7) -- (1,0.7);
		\draw[vote] (0.5,0.1) -- (1,0.1);
\end{scope}
		
\draw[red!85!black,thick] (0.3,0.0) -- (2.4,0.0);
		\draw[red!85!black,thick] (0.3,0.25) -- (2.4,0.25);
		\draw[red!85!black,thick] (0.3,0.5) -- (2.4,0.5);
		\draw[red!85!black,thick] (0.3,0.75) -- (2.4,0.75);
		\draw[red!85!black,thick] (0.3,1.0) -- (2.4,1.0);
		
		\node[text width=1cm, font=\tiny, anchor=west] at (2.45,1.0) {$f_0$};		
		\node[text width=1cm, font=\tiny, anchor=west] at (2.45,0.75) {$f_1$};		
		\node[text width=1cm, font=\tiny, anchor=west] at (2.45,0.5) {$f_2$};		
		\node[text width=1cm, font=\tiny, anchor=west] at (2.45,0.25) {$f_3$};		
		\node[text width=1cm, font=\tiny, anchor=west] at (2.45,0.0) {$f_4$};		
\end{tikzpicture}
\end{subfigure}
	\begin{subfigure}[b]{0.24\textwidth}
		\begin{tikzpicture}
		[yscale=3.5, vote/.style={black!70,thick}]
\draw (0,0) -- (0,1); \draw (-0.05,1) -- (0.05,1); \draw (-0.05,0) -- (0.05,0); \node (y0) at (-0.2,0) {0};
		\node (y1) at (-0.2,1) {1};
		
\begin{scope}[xshift=-0.2cm]
		\fill[fill=black!10] (0.5,0) rectangle (1,1);
		\draw[black!60, thin, fill=white] (0.475,0.315) rectangle (1.025,0.352);
		\node (x1) at (0.75,-0.1) {$j_1$};
		\draw[vote] (0.5,0.3) -- (1,0.3);
		\draw[vote] (0.5,0.45) -- (1,0.45);
		\draw[vote] (0.5,0.21) -- (1,0.21);
		\draw[vote] (0.5,0.1) -- (1,0.1);
\end{scope}
		
		\begin{scope}[xshift=0.6cm]
\fill[fill=black!10] (0.5,0) rectangle (1,1);
		\draw[black!60, thin, fill=white] (0.475,0.48) rectangle (1.025,0.52);
		\node (x1) at (0.75,-0.1) {$j_2$};
		\draw[vote] (0.5,0.5) -- (1,0.5);
		\draw[vote] (0.5,0.05) -- (1,0.05);
		\draw[vote] (0.5,0.1) -- (1,0.1);
		\draw[vote] (0.5,0.81) -- (1,0.81);
\end{scope}
		
		\begin{scope}[xshift=1.4cm]
\fill[fill=black!10] (0.5,0) rectangle (1,1);
		\draw[black!60, thin, fill=white] (0.475,0.48) rectangle (1.025,0.52);
		\node (x1) at (0.75,-0.1) {$j_3$};
		\draw[vote] (0.5,0.2) -- (1,0.2);
		\draw[vote] (0.5,0.5) -- (1,0.5);
		\draw[vote] (0.5,0.7) -- (1,0.7);
		\draw[vote] (0.5,0.1) -- (1,0.1);
\end{scope}
		
\draw[red!85!black,thick] (0.3,0.0) -- (2.4,0.0);
		\draw[red!85!black,thick] (0.3,0.333) -- (2.4,0.333);
		\draw[red!85!black,thick] (0.3,0.666) -- (2.4,0.666);
		\draw[red!85!black,very thick] (0.3,1.0) -- (2.4,1.0);
		
		\node[text width=1cm, font=\tiny, anchor=west] at (2.45,1.0) {$f_0, f_1$};		
		\node[text width=1cm, font=\tiny, anchor=west] at (2.45,0.666) {$f_2$};		
		\node[text width=1cm, font=\tiny, anchor=west] at (2.45,0.333) {$f_3$};		
		\node[text width=1cm, font=\tiny, anchor=west] at (2.45,0.0) {$f_4$};		
\end{tikzpicture}
\end{subfigure} \caption{A moving phantom mechanism operating on an instance with $n=4$ and $m=3$.
}
\label{fig:moving-phantom-example}
\end{figure*}

To build intuition, we consider an example moving phantom mechanism, shown in Figure~\ref{fig:moving-phantom-example}. There are three alternatives, each occupying a column on the horizontal axis, and four voters. Voter reports are indicated by gray horizontal line segments, with their magnitude $\phat_{i,j}$ indicated by their vertical position. The phantom placements are indicated by the red lines and labeled $f_0, \ldots, f_4$. For each alternative, the median of the four agent reports and the five phantoms is indicated by a rectangle.

The four snapshots shown in Figure~\ref{fig:moving-phantom-example} display increasing values of $t$. Observe that the position of each phantom (weakly) increases from left to right, as does the median on each alternative. Although the vertical axis is not labeled, for simplicity of presentation, normalization here occurs in the second image from the left. In the leftmost image, the sum of the highlighted entries is less than 1, while in the two rightmost images it is more than 1.

For simplicity, the definition of moving phantom mechanisms treats the number of voters $n$ as fixed.  To allow $n$ to vary, it is necessary to define a family of phantom systems, one for each $n$.  In the next two sections, we give two examples of such families, but for this section we keep the presentation simple by considering only a fixed $n$. On the other hand, for each fixed $n$, a phantom system defines a mechanism that works for any number $m$ of projects.

Moving phantom mechanisms satisfy some important basic properties. They are all anonymous, neutral, and continuous. Here neutrality is a design choice---one could imagine defining moving phantom mechanisms for which the movement of the phantoms depends on the alternative. 

We now show our main result in this section, that every moving phantom mechanism is incentive compatible. Before proving the result formally, we provide some intuition. If $i$ changes her report from $\mathbf{p}_i$ to $\mathbf{\phat}_i$, the effect on the aggregate can be decomposed into two parts. First, we can think of holding the phantoms fixed at the snapshot dictated by the truthful instance, while changing $i$'s report to $\mathbf{\phat}_i$. Second, we can think of repositioning the phantoms to the snapshot required to guarantee normalization of the aggregate vector after $i$ reports $\mathbf{\phat}_i$. To prove incentive compatibility, we show that any change that the aggregate division undergoes in the first stage can only be bad for voter $i$, pushing the aggregate away from $\mathbf{p}_i$. Change in the second stage can push the aggregate towards $\mathbf{p}_i$, helping voter $i$, but the magnitude of this change is upper bounded by the magnitude of the harmful change in the first stage.

\begin{theorem}
Every moving phantom mechanism is incentive compatible.
\label{thm:ic}
\end{theorem}

\begin{proof}
	Let $\mathcal{F}$ define a moving phantom mechanism $\mathcal{A}^\mathcal{F}$. Consider some report $\mathbf{\phat}_i \neq \mathbf{p}_i$, and fix the reports of all other voters $\mathbf{P}_{-i}$. Let $t^*$ determine the phantom placement for reports $(\mathbf{p}_i,\mathbf{P}_{-i})$ and $\hat{t}^*$ for reports $(\mathbf{\phat}_i,\mathbf{P}_{-i})$. 

	Consider the effect of $i$'s misreport from $\mathbf{p}_i$ to $\mathbf{\phat}_i$ while \emph{holding the phantoms fixed at $\mathcal{F}(t^*)$}. Then, because phantom placements are fixed on each alternative, any change that voter $i$ can cause on alternative $j$ by misreporting must be away from her ideal value $p_{i,j}$. For each $j\in [m]$, we have
	\begin{alignat*}{4}
		\med(\mathcal{F}(t^*),\phat_{i,j}, \mathbf{P}_{-i,j}) &\ge \med(\mathcal{F}(t^*), p_{i,j}, \mathbf{P}_{-i,j}) &\qquad& 
		\text{if } p_{i,j} \le \mathcal{A}^\mathcal{F}(\mathbf{p}_i,\mathbf{P}_{-i})_j < \phat_{i,j}, \\
		\med(\mathcal{F}(t^*),\phat_{i,j}, \mathbf{P}_{-i,j}) &\le \med(\mathcal{F}(t^*),p_{i,j}, \mathbf{P}_{-i,j}) && 
		\text{if } \phat_{i,j} \le \mathcal{A}^\mathcal{F}(\mathbf{p}_i,\mathbf{P}_{-i})_j < p_{i,j}, \\
		\med(\mathcal{F}(t^*),\phat_{i,j}, \mathbf{P}_{-i,j}) &= \med(\mathcal{F}(t^*),p_{i,j}, \mathbf{P}_{-i,j}) &&
		\text{otherwise.}
	\end{alignat*}

	Let $y_j = \med(\mathcal{F}(t^*),\phat_{i,j}, \mathbf{P}_{-i,j})-\med(\mathcal{F}(t^*),p_{i,j}, \mathbf{P}_{-i,j})$ denote the change caused on alternative~$j$ by voter $i$'s misreport, subject to holding the phantom placement fixed at $\mathcal{F}(t^*)$. By the above,
	\begin{equation}
		\label{eqn:proof-part-1}
		\sum_{j \in [m]} |p_{i,j}-\med(\mathcal{F}(t^*),\phat_{i,j}, \mathbf{P}_{-i,j})|=\sum_{j \in [m]} |p_{i,j}-\mathcal{A}^\mathcal{F}(\mathbf{p_i}, \mathbf{P}_{-i})_j| + \sum_{j \in [m]} |y_j|.
	\end{equation}

	Next, we consider the change that results from moving the phantoms from $\mathcal{F}(t^*)$ to $\mathcal{F}(\hat{t}^*)$.
	Assume that $\sum_{j \in [m]} y_j \ge 0$ (otherwise, a similar argument applies). Then $\sum_{j \in [m]} \med(\mathcal{F}(t^*),\phat_{i,j}, \mathbf{P}_{-i,j}) \ge 1$, which implies that $\hat{t}^* \le t^*$ since the sum is monotonic in $t$ (see the proof of Proposition~\ref{prop:existence}).
	Thus, for the aggregate division $\mathcal{A}^\mathcal{F}(\mathbf{\phat}_i,\mathbf{P}_{-i})$ we have that 
	\[
	  \mathcal{A}^\mathcal{F}(\mathbf{\phat}_i,\mathbf{P}_{-i})_j= \med(\mathcal{F}(\hat{t}^*),\phat_{i,j}, \mathbf{P}_{-i,j}) \le \med(\mathcal{F}(t^*),\phat_{i,j}, \mathbf{P}_{-i,j})
	   \quad\text{for all $j\in[m]$.}
	\] 
	Therefore,
	\begin{align}
		\sum_{j \in [m]} |\med(\mathcal{F}(t^*),\phat_{i,j}, \mathbf{P}_{-i,j})-\mathcal{A}^\mathcal{F}(\mathbf{\phat}_i,\mathbf{P}_{-i})_j| &= \sum_{j \in [m]} (\med(\mathcal{F}(t^*),\phat_{i,j}, \mathbf{P}_{-i,j})-\mathcal{A}^\mathcal{F}(\mathbf{\phat}_i,\mathbf{P}_{-i})_j) \nonumber \\
		&= \sum_{j \in [m]} (\med(\mathcal{F}(t^*),\phat_{i,j}, \mathbf{P}_{-i,j}))-\sum_{j \in [m]} \mathcal{A}^\mathcal{F}(\mathbf{\phat}_i,\mathbf{P}_{-i})_j \nonumber \\
		&= \sum_{j \in [m]} (\med(\mathcal{F}(t^*),\phat_{i,j}, \mathbf{P}_{-i,j}))- 1 \nonumber \\
		&= \sum_{j \in [m]} (\med(\mathcal{F}(t^*),\phat_{i,j}, \mathbf{P}_{-i,j}))-\sum_{j \in [m]} \mathcal{A}^\mathcal{F}(\mathbf{p}_i,\mathbf{P}_{-i})_j \nonumber \\
		&= \sum_{j \in [m]} (\med(\mathcal{F}(t^*),\phat_{i,j}, \mathbf{P}_{-i,j})- \mathcal{A}^\mathcal{F}(\mathbf{p}_i,\mathbf{P}_{-i})_j) \nonumber \\
		&= \sum_{j \in [m]} y_j \label{eqn:proof-part-2}
	\end{align}
	We complete the proof by noting that
	\begin{align*}
		d(\mathbf{p}_i, \mathcal{A}^\mathcal{F}(\mathbf{\phat}_i, \mathbf{P}_{-i})) &= \sum_{j \in [m]} |p_{i,j}-\mathcal{A}^\mathcal{F}(\mathbf{\phat}_i,\mathbf{P}_{-i})_j|\\
		&\ge \sum_{j \in [m]} |p_{i,j}-\med(\mathcal{F}(t^*),\phat_{i,j}, \mathbf{P}_{-i,j})|\\&\qquad \qquad -\sum_{j \in [m]} |\mathcal{A}^\mathcal{F}(\mathbf{\phat}_i,\mathbf{P}_{-i})_j-\med(\mathcal{F}(t^*),\phat_{i,j}, \mathbf{P}_{-i,j})|\\
		&= \sum_{j \in [m]} |p_{i,j}-\mathcal{A}^\mathcal{F}(\mathbf{p_i}, \mathbf{P}_{-i})_j| + \sum_{j \in [m]} |y_j| - \sum_{j \in [m]} y_j\\
		&\ge \sum_{j \in [m]} |p_{i,j}-\mathcal{A}^\mathcal{F}(\mathbf{p_i}, \mathbf{P}_{-i})_j| = d(\mathbf{p}_i, \mathcal{A}^\mathcal{F}(\mathbf{p_i}, \mathbf{P}_{-i})),
	\end{align*}
	where the first and last equalities follow from the definition of the $\ell_1$ metric, the first inequality holds by an application of the triangle inequality, the second equality follows from Equations~\ref{eqn:proof-part-1} and~\ref{eqn:proof-part-2}, and the final inequality holds because $|x| \ge x$ for all $x \in \mathbb{R}$.
\end{proof}

In addition to incentive compatibility, moving phantom mechanisms satisfy a natural monotonicity property that says that if some voter increases her report on alternative $j$, and decreases her report on all other alternatives, then the aggregate weight on alternative $j$ should not decrease.

\begin{definition}
	A budget aggregation mechanism $\mathcal{A}$ satisfies \emph{monotonicity} if, for all $\mathbf{p}_i$, $\mathbf{p}'_i$ with $p_{i,j} > p'_{i,j}$ for some $j$ and $p_{i,k} \le p'_{i,k}$ for all $k \neq j$, 
	\[ \mathcal{A}(\mathbf{p}_i,\mathbf{P}_{-i})_j \ge \mathcal{A}(\mathbf{p}'_i,\mathbf{P}_{-i})_j. \]
\end{definition}

\begin{theorem}
	Every moving phantom mechanism satisfies monotonicity.
	\label{thm:monotonicity}
\end{theorem}

 \begin{proof}
 	Let $\mathbf{p}_i$, $\mathbf{p}'_i$ be such that $p_{i,j} > p'_{i,j}$ for some $j$ and $p_{i,k} \le p'_{i,k}$ for all $k \neq j$. Let $t^*$ determine the phantom placement for reports $(\mathbf{p}_i,\mathbf{P}_{-i})$ and $t'^*$ for reports $(\mathbf{p}'_i,\mathbf{P}_{-i})$.

 	Suppose that $t'^* < t^*$. We have
 		\[ \mathcal{A}(\mathbf{p}'_i,\mathbf{P}_{-i})_j = \med (\mathcal{F}(t'^*),p'_{i,j}, P_{-i,j})
 		\le \med (\mathcal{F}(t^*),p_{i,j}, P_{-i,j})
 		= \mathcal{A}(\mathbf{p}_i,\mathbf{P}_{-i})_j \]
 	where the inequality holds because $p_{i,j} > p'_{i,j}$ and $f_k(t^*) \ge f_k(t'^*)$ for all $k \in \{ 0, \ldots, n \}$.

 	Next, suppose that $t'^* > t^*$. Then
 	\begin{align*}
 		\mathcal{A}(\mathbf{p}'_i,\mathbf{P}_{-i})_j = 1-\sum_{j' \neq j} \mathcal{A}(\mathbf{p}'_i,\mathbf{P}_{-i})_{j'}
 		&= 1-\sum_{j' \neq j} \med (\mathcal{F}(t'^*),p'_{i,j'}, P_{-i,j'})\\
 		&\le 1-\sum_{j' \neq j} \med (\mathcal{F}(t^*),p_{i,j'}, P_{-i,j'})
 		= \mathcal{A}(\mathbf{p}_i,\mathbf{P}_{-i})_j
 		\end{align*}
 	where the inequality holds as $p_{i,j'} < p'_{i,j'}$ for all $j' \neq j$ and $f_k(t^*) \le f_k(t'^*)$ for all $k \in \{ 0, \ldots, n \}$.
 \end{proof}

Let us briefly discuss how to compute the output of a moving phantom mechanism. In principle, the precise time $t^* \in [0,1]$ at which the output of the mechanism is normalized may have many decimal digits, and for badly-behaved $\mathcal{F}$ it may even be irrational. For the same reason, the mechanism may produce an irrational division, so the precise computation of the output may not be possible. However, for mechanisms for which the phantom system $\mathcal{F}$ is piecewise linear, we can show that $t^*$ has few digits and the output is always rational, so polynomial-time computation is possible using a binary search on $t$. A phantom system is \emph{piecewise linear} if we can partition time into $r$ time intervals, such that within each interval, all the phantoms move at constant speed. All the specific moving phantom mechanisms that we consider in the remainder of the paper have a piecewise linear phantom system, and can thus be efficiently evaluated.

\begin{theorem}
	\label{thm:rational}
	Suppose that $\mathcal F$ is a piecewise linear phantom system, so that there are rationals $0 = t_1 < t_2 < \cdots < t_r = 1$ such that for each $k = 0, \dots, n$ and each $s \in [r-1]$, there are rational coefficients $a_k^s \in \mathbb Q$ and $b_k^s \in \mathbb Q$ with
	\[
		f_k(t) = a_k^s \cdot t + b_k^s \quad \text{whenever $t_s \le t \le t_{s+1}$}.
	\]
	Let $\mathbf P$ be a profile in which all reports $p_{i,j}$ are rational. Then there exists a rational time $t^* \in \mathbb Q$ for which $\mathcal A^{\mathcal F}$ is normalized, and $t^*$ can be specified in polynomially many bits. Hence, the outcome distribution of $\mathcal A^{\mathcal F}$ is rational and can be computed in polynomial time.
\end{theorem}
\begin{proof}
	We show that $t^*$ and the outcome distribution $p_1, \dots, p_m$ form the solution of a system of linear inequalities with rational coefficients. Since a system of linear inequalities with rational coefficients that has a solution always has a rational solution of polynomial size \citep[Thm.~10.1]{Schr86a}, the first claim of the theorem follows. For the second claim about computability, note that given a (rational) time $t$, we can clearly compute the outcome of $\mathcal A^{\mathcal F}$ at time $t$, and we can use a binary search on $t$ to find the time $t^*$ when $\mathcal A^{\mathcal F}$ is normalized. The binary search takes polynomially many steps since $t^*$ is guaranteed to have polynomial size. (In this analysis of time complexity, the number of voters $n$ is fixed, but the number $m$ of alternatives is allowed to vary.)
	
	We now construct the mentioned system of linear inequalities. Let $t^*$ be a time when the mechanism is normalized, 
	and choose
	$s \in [r-1]$ such that $t_s \le t^* \le t_{s+1}$. Our system of linear inequalities will have variables $t, p_1, \dots, p_m \ge 0$. We add inequalities (also called constraints) to the system one by one. First, we add the constraints $\sum_{j\in[m]} p_j = 1$ and $t_s \le t \le t_{s+1}$. Next, we will add constraints that, for each alternative $j \in [m]$, encode $p_j$ as the value that is the median for $j$.
	
	Let $j \in [m]$. Suppose first that $\mathcal A^{\mathcal F}(\mathbf{P})_j$ is a voter report, i.e. $\mathcal A^{\mathcal F}(\mathbf{P})_j = p_{i,j}$ for some $i \in N$. Then we add the constraint $p_j = p_{i,j}$. Then, for each phantom $k = 0, \dots, n$, we add the constraint
	\begin{align*}
		a_k^s \cdot t + b_k^s \le p_j & \quad \text{if $f_k(t^*) \le \mathcal A^{\mathcal F}(\mathbf{P})_j$,} \\
		a_k^s \cdot t + b_k^s \ge p_j & \quad \text{if $f_k(t^*) \ge \mathcal A^{\mathcal F}(\mathbf{P})_j$,}
	\end{align*}
	which implies that $p_j$ is the median of $p_{1,j}, \dots, p_{n,j}, f_0(t), \dots, f_n(t)$.
	
	Otherwise, we have that $\mathcal A^{\mathcal F}(\mathbf{P})_j$ is not a voter report, so $\mathcal A^{\mathcal F}(\mathbf{P})_j = f_{k'}(t^*)$ for some $k' \in \{0, \dots, n\}$. Then again for each phantom $k = 0, \dots, n$, we add the constraint
	\begin{align*}
		a_k^s \cdot t + b_k^s \le p_j & \quad \text{if $f_k(t^*) \le \mathcal A^{\mathcal F}(\mathbf{P})_j$,} \\
		a_k^s \cdot t + b_k^s \ge p_j & \quad \text{if $f_k(t^*) \ge \mathcal A^{\mathcal F}(\mathbf{P})_j$,}
	\end{align*}
	which implies that $p_j = a_{k'}^s \cdot t + b_{k'}^s$, and that $p_j$ is the median of $p_{1,j}, \dots, p_{n,j}, f_0(t), \dots, f_n(t)$.
	
	Clearly, any solution to this system of constraints encodes an outcome of $\mathcal A^{\mathcal F}$, as required.
\end{proof}

Before we move on to particular moving phantom mechanisms, let us end this section with a tantalizing open question: Does there exist an (anonymous, neutral, continuous) incentive-compatible budget aggregation mechanism that is \emph{not} a moving phantom mechanism? We have not been able to construct any example, and have found that some mechanisms that on first sight seem to have nothing to do with medians end up having an equivalent description as a moving phantom mechanism. For the simpler two-alternative case, we already have a characterization of all incentive-compatible mechanisms (Theorem~\ref{thm:binary}). This class can equivalently be described in terms of moving phantoms, and so the answer to our question for $m=2$ is \emph{no}.

\begin{theorem}
	For $m=2$, moving phantom mechanisms are the only budget aggregation mechanisms that satisfy anonymity, neutrality, continuity, and incentive compatibility.
\end{theorem}

\begin{proof}
	Certainly all moving phantom mechanisms satisfy these properties.
	For the other direction, we know from Theorem~\ref{thm:binary} that any mechanism $\mathcal{A}$ satisfying these properties can be described as a generalized median with phantoms $\alpha_0, \dots, \alpha_n$ satisfying, due to neutrality, $ \{ \alpha_0, \dots, \alpha_{n} \} = \{ 1-\alpha_0, \dots, 1-\alpha_{n} \}$.
	We show that $\mathcal{A}$ is equivalent to a moving phantom mechanism.
	Define $\mathcal{A}^\mathcal{F}$ using a phantom system $\mathcal{F}$ for which there exists a $t^* \in [0,1]$ with $f_k(t^*)=\alpha_k$ for every $k \in \{ 0, \ldots, n \}$. 
	Then, for every preference profile $\mathbf{P}$, we have that $\mathcal{A}^\mathcal{F}(\mathbf{P})_1 = \med (\mathcal{F}(t^*), \mathbf{P}_{i \in [n],j}) = (\alpha_0, \ldots, \alpha_n, \mathbf{P}_{i \in [n],j})$, and $\mathcal{A}^\mathcal{F}(\mathbf{P})_2 = 1-\mathcal{A}^\mathcal{F}(\mathbf{P})_1$, matching the output of $\mathcal{A}$.
\end{proof}
 \section{The Independent Markets Mechanism}
\label{sec:im}

We have seen that uniform phantoms is uniquely proportional for $m=2$. By a similar argument to the proof of Proposition~\ref{prop:uniform-proportional}, it is easy to see that any family of functions $\mathcal{F}$ that generates uniform phantoms at some snapshot will be proportional, and will reduce to the uniform phantom mechanism for $m=2$. However, this leaves a large class of moving phantom mechanisms to choose from. In this section, we identify a particular moving phantom mechanism that generalizes the uniform phantom mechanism for arbitrary $m$. Its outcome can be interpreted as a market equilibrium, which will be useful for analyzing the mechanism.

\begin{definition}
	\label{def:im}
	The independent markets mechanism ($\mathcal{A}^{\IM}$) is the moving phantom mechanism defined by the phantom system 
	$f_k(t) = \min \{ t(n-k), 1 \}$
for each $k \in \{ 0, \ldots, n \}$.
\end{definition}

To visualize the phantom placement, observe that for any $t \leq 1/n$, phantoms are being placed at $0, t, 2t, \ldots, nt$. Once $t$ reaches $1/n$, phantoms continue to grow in the same manner, but the higher phantoms get capped at 1.\footnote{As written, $f_n(1)=0$, but Definition~\ref{def:moving-phantom} requires $f_k(1)=1$ for all $k$. This detail does not matter here, since normalization is always achieved without moving phantom $n$, but one could write $f_n$ in a different form to satisfy Definition~\ref{def:moving-phantom} without it changing the behavior of the mechanism.} This is actually the mechanism that we displayed in Figure~\ref{fig:moving-phantom-example}.
Note that, when $t = 1/n$, the phantom placement is uniform on $[0,1]$ (as is the case in the third panel of Figure~\ref{fig:moving-phantom-example}); thus, $\mathcal{A}^{\IM}$ reduces to the uniform phantom mechanism for $m=2$.

\begin{example}
	\label{ex:im}
	Consider the profile $\mathbf P$ from Example~\ref{ex:uniform-not-normalized} with voter reports $\mathbf{p}_1=(0, 0.5, 0.5), \mathbf{p}_2=(\frac{1}{3}, \frac{2}{3}, 0)$, and $\mathbf{p}_3=(0.9,0,0.1)$. Consider the placement of the $n+1=4$ phantoms when $t=\frac{2}{9}$. They are placed at $f_0(t)=\frac{6}{9}, f_1(t)=\frac{4}{9}, f_2(t)=\frac{2}{9}, f_3(t)=0$. The generalized medians are
	\begin{align*}
		\med \{ f_0(t), f_1(t), f_2(t), f_3(t), p_{1,1}, p_{2,1}, p_{3,1} \} &= \textstyle \med \{ \frac{6}{9}, \frac{4}{9}, \frac{2}{9}, 0, 0, \frac{1}{3}, 0.9 \}=\frac{1}{3}, \\
		\med \{ f_0(t), f_1(t), f_2(t), f_3(t), p_{1,2}, p_{2,2}, p_{3,2} \} &= \textstyle \med \{ \frac{6}{9}, \frac{4}{9}, \frac{2}{9}, 0, 0.5, \frac{2}{3}, 0 \}=\frac{4}{9}, \\
		\med \{ f_0(t), f_1(t), f_2(t), f_3(t), p_{1,3}, p_{2,3}, p_{3,3} \} &= \textstyle \med \{ \frac{6}{9}, \frac{4}{9}, \frac{2}{9}, 0, 0.5, 0, 0.1 \}=\frac{2}{9}.
	\end{align*}
	Because $\frac13 + \frac49+\frac29 = 1$, these are normalized, and hence we have $t^*=\frac29$. Thus $\mathcal{A}^{\IM}(\mathbf{P})=(\frac{1}{3},\frac{4}{9},\frac{2}{9})$.
\end{example}

\subsection{Market Interpretation}

Why do we call this mechanism the independent markets mechanism? To explain this, we first establish a connection between the market clearing price in a simple single-good market and the median of some familiar-looking numbers.

Suppose we are selling a single divisible good, of which a total amount of $x\in [0, \infty)$ is available.
Each of $n$ voters has a budget of 1, and a value $v_i \in [0, \infty)$ per unit of the good. 
We assume voters have quasilinear utilities $u_i(x_i) = x_i \cdot (v_i - \pi)$, where $\pi \ge 0$ is the price per unit of the good. 
For a price $\pi$, the demand of voter $i$, $D_i(\pi) \in [0,\infty]$ must satisfy
\begin{alignat*}{3}
	&D_i(0)=\infty &&\\
	&D_i(\pi)=1/\pi &&\text{ for all }0<\pi<v_i\\
	&D_i(v_i) \in [0,1/v_i] &&\\
	&D_i(\pi)=0 &&\text{ for all }\pi>v_i
\end{alignat*}
Thus, each voter demands as much of the good as their budget of 1 allows at price $\pi$, as long as the price per unit is lower than their value per unit. 
When the price is equal to the voter's value, we allow the demand to be anywhere between $0$ and $\frac{1}{\pi}$, since the voter is indifferent between purchasing any amount of the good at this price, subject to her budget.
The \emph{market clearing price} $c$ is the price at which the supply of the good ($x$) equals the total demand. Formally, we say that $c$ is a market clearing price if there exist demands $D_i(c)$ such that
\begin{equation}
	\textstyle\sum_{i \in [n]} D_i(c) = x.
\end{equation}

It turns out that a unique market clearing price $c$ always exists and is equal to the median of the $n$ voter values $v_i$ and of $n+1$ ``phantom values'' which are uniformly distributed on the interval $[0, n/x]$.
To the best of our knowledge, this connection has not previously been appreciated in the literature.

\begin{lemma}
	\label{thm:market-median}
	In the market defined above, a unique market clearing price $c$ always exists and equals 
\[\med (0, 1/x, \ldots, (n-1)/x, n/x, v_1, \ldots, v_n).\]
\end{lemma}
\begin{proof}
	We distinguish the cases that the median is a phantom $a/x$ or a voter report $v_i$.
	
	Suppose that the median is $a/x$ for some $a \in \{0, \dots, n\}$. Then we can partition the (phantom and voter report) entries, with the exception of the phantom at $a/x$, into sets $A$ and $B$ with $|A|=|B|=n$, where $A$ consists only of entries less than or equal to $a/x$, and $B$ consists only of entries greater than or equal to $a/x$. Note that $B$ contains $n-a$ phantoms, and thus $n-(n-a)=a$ voter reports.
		
	We show that $\pi = a/x$ is a market clearing price. 
	At this price, we can set the demand of each voter $i \in N$ whose report is in $B$ to $D_i(\pi) = 1/\pi = x/a$, and we can set the demand of each voter whose report is not in $B$ to 0. Since $B$ contains $a$ voter reports, the total demand of all voters is $a \cdot x/a = x$, and hence $\pi = a/x$ is a market clearing price. To see that it is the unique clearing price, note that at any price $\pi<a/x$, each voter whose report is in $B$ has demand $D_i(\pi) = 1/\pi > x/a$, so the total demand of all voters in $B$ is greater than $x$. At any price $\pi>a/x$, each voter whose report is in $A$ has demand $D_i(\pi)=0$, while each of the $a$ voters whose report is in $B$ has demand at most $D_i(\pi)=1/\pi<x/a$, so the total demand of all voters is less than $x$.
	
	Next, suppose that the median $y$ is not a phantom entry, and say that we have $a/x < y < (a+1)/x$ for some $a \le n-1$ (note that the median cannot be greater than $n/x$, because it cannot be higher than the largest phantom value). Then we can partition the (phantom and voter report) entries, not including a single voter with $v_i=y$ (one such voter must exist because the median coincides with some entry, and no phantom entry lies at $y$), into sets $A$ and $B$ each of size $n$, where $A$ consists only of entries less than or equal to $y$, and $B$ consists only of entries greater than or equal to $y$. Again, $B$ contains $n-a$ phantom reports, and thus $a$ voter reports. 
	
	We show that $\pi = y$ is a market clearing price. 
	For each voter $i$ with $v_i=y$ whose report is in $B$, let $D_i(v_i)=1/v_i=1/y$. For each voter $i$ with $v_i=y$ whose report is in $A$, let $D_i(\pi)=0$. Finally, for the single voter $i$ with $v_i=y$ whose report is neither in $A$ nor in $B$ , set $D_i(v_i)=x-a/y$. Note that $D_i(v_i) \in [0, 1/v_i]$ because $a/x<y$ implies $x-a/y>0$ and $y<(a+1)/x$ implies $x-a/y<1/y$. Therefore, these settings are consistent with the requirements of a demand function.
	At price $\pi=y$, each of the $a$ voters whose report is in $B$ has demand $D_i(y)=1/y$, the voter whose report is in neither $A$ nor $B$ has demand $D_i(y)=x-a/y$, and all other voters have demand 0, so the total demand is exactly $x$. Hence, $\pi = y$ is a market clearing price.
	To see that it is the unique clearing price, note that at all prices $\pi<y$, each voter $i$ whose report is in $B$, as well as the voter whose report is in neither $A$ nor $B$ has demand $1/\pi > 1/y$. The total demand is thus greater than $(a+1) \cdot 1/y > x$. Moreover, at all prices $\pi>y$, all voters $i$ whose report is not in $B$ have demand $D_i(\pi)=0$, and all voters $i$ whose report is in $B$ have demand $D_i(\pi) \le 1/\pi<1/y$. The total demand is thus less than $a\cdot 1/y<x$. Therefore, $y$ is the unique clearing price.
\end{proof}

The ``market'' connection to independent markets is now clear: For each alternative $j$, we set up a market in which we sell an amount $x$ of a good; this amount is the same across markets. Voter $i\in[n]$ has value $\phat_{i,j}$ for the good sold in market $j$, and has a budget of 1 in each market. The markets are ``independent'' because, while each voter is engaged in every market, the budget of 1 for each market can only be used to buy the good sold in that market. Using Lemma~\ref{thm:market-median}, we can derive the market clearing prices in each of these markets. If we write $t = 1/x$, then these prices correspond exactly to the output of $\mathcal{A}^{\IM}$ with the phantoms as placed at time $t$. Changing the phantom placement by varying $t$ to normalize the output is equivalent to varying the amount $x$ of the good sold in each market until the clearing prices across markets sum to 1. While we prevent phantoms from moving above 1 in the definition of independent markets---complying with Definition~\ref{def:moving-phantom}---the exact positions of these phantoms do not affect the clearing price since all reports are at most 1.

Returning to Example~\ref{ex:im}, we can verify the outcome using the market interpretation, by setting the quantity of goods to be sold in each market to $x^*=1/t^*=4.5$. In the market corresponding to alternative 1, the market clears at price $\pi_1=\frac{1}{3}$, at which price voter 3 demands $1/\pi_1=3$ goods, voter 1 demands nothing, and the demand of voter 2 can be set to $D_2(\frac{1}{3})=1.5$, which lies between 0 and $1/v_i=3$. Supply therefore equals demand.
It can be checked that the market prices also match the independent markets outcome for alternatives 2 and 3.

An interesting property of independent markets is that normalization can always be achieved with $x^* \ge n$, that is, a supply of at least $n$ units of each good.
\begin{proposition}
	\label{prop:uniform-overnormalizes}
Suppose we run independent markets with supply $x \le n$. Then the clearing prices sum to at least 1. 
\end{proposition}
\begin{proof}
	Let $\pi_1, \dots, \pi_m$ be the clearing prices for supply $x$.
	We may assume that for each agent $i \in [n]$, there exists a project $j$ for which $\pi_j < \hat p_{i,j}$, since otherwise we have $\sum_{j \in [m]} \pi_j \ge \sum_{j \in [m]} \hat p_{i,j} = 1$ and we are done. Thus, for each agent there exists at least one market where the agent spends all  available funds of 1 (by definition of demand). Since each agent spends at least 1, total spending across all markets is at least $n$.
	Now, the cost of purchasing all available goods is $x \cdot \sum_{j \in [m]}\pi_j$. Since the markets clear, this is the overall amount of money that the agents spend. Hence $x \cdot \sum_{j \in [m]}\pi_j \ge n$. Since $x \le n$, we obtain $\sum_{j \in [m]}\pi_j \ge 1$, as desired.
\end{proof}
 In the moving phantom interpretation, Proposition~\ref{prop:uniform-overnormalizes} means that $t^* \le \frac1n$ will always suffice for achieving normalization, so that the phantom movement will never need to continue past the uniform phantom placement $\mathcal F^{\IM}(\frac1n) = (1, \frac{n-1}{n}, \dots, \frac1n, 0)$. 
 In turn, this shows that if we used the uniform phantom mechanism on each project independently (as discussed in Example~\ref{ex:uniform-not-normalized}) the outcome will always sum to at least 1. The independent markets mechanism always returns a distribution $\mathbf p$ where each coordinate is at most the result of taking uniform phantoms in each project independently.

The market system we have described yields an incentive-compatible aggregator, since it corresponds to a moving phantom mechanism. There are other market-based aggregation mechanisms described in the literature, most famously the parimutuel consensus mechanism of \citet{EG59}. That mechanism differs from ours in that voters have only a single budget of 1 which they can use in all of the markets.
(The supply of goods can be fixed at $x=n$, which guarantees that prices are normalized, because total spending is fixed.)
For the case $m=2$, it does not matter whether markets are independent or not, and our mechanism is equivalent to the one of \citet{EG59}. It follows that the parimutuel consensus mechanism is incentive compatible for $m=2$ (in our $\ell_1$ sense). However, for $m\ge 3$, the mechanism is manipulable,\footnote{Let $\mathbf{p}_1=(0,0.5,0.5)$, $\mathbf{p}_2=(0.5,0.5,0)$. Parimutuel consensus yields prices $(1/3,1/3,1/3)$, at distance $2/3$ from $\mathbf{p}_1$. If voter 1 instead reports $\mathbf{\phat}_1=(0,0,1)$, the price vector is $(0.25,0.25,0.5)$, at distance $0.5$ from $\mathbf{p}_1$. This is an example of an agent manipulating prices in a competitive market, which has been studied in more general settings \citep[e.g.,][]{roberts1976incentives}.}
 and hence cannot be represented as a moving phantom mechanism.
We point the reader to the work of \citet{GGP18} for a detailed overview of other settings in which market mechanisms have been used in the context of public decision making.

\subsection{Voting Game Interpretation}

We have seen descriptions of $\mathcal{A}^{\IM}$ as a moving phantom mechanism and as clearing prices of a market system. We next give a game-theoretic description: the independent markets mechanism can be seen as the unique Nash equilibrium outcome of a voting game inspired by range voting~\citep{S00}. The game works as follows. Each voter receives one ``credit'' per alternative, and chooses how much of that credit to place on the alternative. That is, each voter chooses a vector $\mathbf{s}_i \in [0,1]^m$. The outcome of the game is the division $\mathbf p$ where $p_j$ is proportional to $\sum_{i \in [n]} s_{i,j}$, the number of credits spent on alternative $j$. Voters choose their spending so as to obtain an outcome that is as close to their ideal point $\mathbf\phat_i$ as possible, according to $\ell_1$ distance.

\begin{theorem}
The voting game defined above has a unique outcome $\mathbf p$ that can be obtained in Nash equilibrium, and it is equal to the outcome of the independent markets mechanism.
\end{theorem}

The idea of the proof is to set $\mathbf{s}_{i,j}$ equal to the amount that agent $i$ spends in the market for alternative $j$, under the market setup that we described earlier. We show that when every agent casts ``vote'' $\mathbf{s}_i$, the system is at (its unique) equilibrium, and that the spending is proportional to the independent markets outcome.

\begin{proof}
	Let $\mathbf{s}_{-i}$ denote the spending profile of all voters other than $i$, and let $\mathbf{q}$ denote the aggregate division in which the weight on an alternative is proportional to the number of credits spent on that alternative. Suppose that $\mathbf{s}_i$ is a best response for voter $i$. We first show that for every alternative $j$, $q_j<p_{i,j}$ implies $s_{i,j}=1$ and $q_j>p_{i,j}$ implies $s_{i,j}=0$ (we only prove the former, but the latter follows from a symmetric argument).
That is, voters will always prefer to increase their spending on alternatives that they consider ``undervalued'', and decrease their spending on ``overvalued'' alternatives.
	
	To prove this, assume for a contradiction that $q_j < p_{i,j}$ and $s_{i,j}<1$. 
	Since $\mathbf{q}$ and $\mathbf{p}_i$ are both normalized, there must exist some alternative $j'$ for which $q_{j'}>p_{i,j'}$. 
	By marginally increasing $s_{i,j}$ while holding $s_{i,\hat{j}}$ constant for all $\hat{j} \neq j$, voter $i$ can move the aggregate division from $\mathbf{q}$ to $\tilde{\mathbf{q}}$, where $\tilde{q}_{j} = q_j+\epsilon < p_{i,j}$, and $\tilde{q}_{\hat{j}} = q_{\hat{j}} - \epsilon_{\hat{j}}$ for all $\hat{j} \neq j$, with $\sum_{\hat{j} \neq j} \epsilon_{\hat{j}} = \epsilon$ and $\tilde{q}_{j'}>p_{i,j'}$. This implies that
	\begin{align*}
d(\mathbf{q},\mathbf{p}_i) - d(\tilde{\mathbf{q}},\mathbf{p}_i) 
&= |q_j - p_{i,j}| - |\tilde{q}_j-p_{i,j}| + |q_{j'} - p_{i,j'}| - |\tilde{q}_{j'}-p_{i,j'}| + \sum_{\hat{j} \neq j,j'} (|q_{\hat{j}} - p_{i,\hat{j}}| - |\tilde{q}_{\hat{j}}-p_{i,\hat{j}}|)\\
&\ge \epsilon + \epsilon_{j'} - \sum_{\hat{j} \neq j,j'} \epsilon_{\hat{j}}
= 2 \epsilon_{j'} >0
\end{align*}
and hence $d(\mathbf{q},\mathbf{p}_i) > d(\tilde{\mathbf{q}},\mathbf{p}_i)$, contradicting that $\mathbf{s}_i$ is a best response for $i$.

	Next, we show that there exists an equilibrium of the voting game that produces the same outcome as independent markets. To this end, consider the market interpretation. For every alternative, some amount $x^*=1/t^*$  of a divisible good is sold to voters with a budget of 1 each. The amount of money that each voter spends in each market defines some spending profile $\mathbf{s}$.\footnote{This spending profile is not unique, since for voters with $p_{i,j}$ equal to the clearing price for alternative $j$, there is some flexibility as to which voters pay for and get assigned goods, and which do not.} Importantly, whenever $p_{i,j}>\mathcal{A}^{\IM}(\mathbf{P})_j$, voter $i$ spends their full budget on alternative $j$ (i.e. $s_{i,j}=1$), and whenever $p_{i,j}<\mathcal{A}^{\IM}(\mathbf{P})_j$, voter $i$ spends nothing on alternative $j$ (i.e. $s_{i,j}=0$). Further, the total amount spent by all voters on alternative $j$ is $x^* \cdot \mathcal{A}^{\IM}(\mathbf{P})_j$ -- the amount of goods sold multiplied by the price per unit -- which is proportional to the aggregate division $\mathcal{A}^{\IM}(\mathbf{P})$. Therefore, the induced spending profile does produce the same aggregate division as independent markets when aggregated under the rules of the voting game. We now show that this spending profile is a Nash equilibrium.
	
	To do so, consider the spending vector $\mathbf{s}_i$ of some voter $i$, and suppose they have a better response $\tilde{\mathbf{s}}_i$. Denote the aggregate division when $i$ spends $\mathbf{s}_i$ by $\mathbf{q}$ (this division is $\mathcal{A}^{\IM}(\mathbf{P})$, but we use $\mathbf{q}$ for short), and the aggregate division when $i$ spends $\tilde{\mathbf{s}}_i$ by $\tilde{\mathbf{q}}$. 
	Since $d(\tilde{\mathbf{q}}, \mathbf{p})<d(\mathbf{q}, \mathbf{p})$, there must exist some alternative $j$ with $q_j < p_{i,j}$ and $q_j < \tilde{q}_j$, or with $q_j > p_{i,j}$ and $q_j > \tilde{q}_j$. 
	Suppose without loss of generality that the former case holds (if not, a very similar argument applies). Because $q_j < p_{i,j}$, we have $s_{i,j}=1$ by definition of $\mathbf{s}$. 
	Therefore, the only way for 
	\[ q_j = \frac{s_{i,j}+\sum_{i' \neq i}s_{i',j}}{\sum_{j'} s_{i,j'}+\sum_{i' \neq i} \sum_{j'} s_{i',j'}} < \frac{\tilde{s}_{i,j}+\sum_{i' \neq i}s_{i',j}}{\sum_{j'} \tilde{s}_{i,j'}+\sum_{i' \neq i} \sum_{j'} s_{i',j'}} = \tilde{q}_j \]
	is if we have $\sum_{j'} s_{i,j'} > \sum_{j'} \tilde{s}_{i,j'}$. But, because $s_{i,j'}=0$ for all $j'$ with $q_{j'} > p_{i,j'}$, 
	it must be the case that $q_{j'} \le p_{i,j'}$ for all alternatives with $s_{i,j'}>\tilde{s}_{i,j'}$.
	Further, any alternative with $\tilde{q}_{j'}<q_{j'}$ must have $\tilde{s}_{i,j'}<s_{i,j'}$; this follows from the fact that $\sum_{j'} s_{i,j'} > \sum_{j'} \tilde{s}_{i,j'}$.
	Putting the previous two observations together, we have that $q_{j'} \le p_{i,j'}$  for all alternatives with $\tilde{q}_{j'}<q_{j'}$. Therefore
	\begin{align*}
		d(\mathbf{q},\mathbf{p}_i)-d(\tilde{\mathbf{q}},\mathbf{p}_i) &= \sum_{\tilde{q}_{j'}<q_{j'}} (|q_{j'}-p_{i,j'}| - |\tilde{q}_{j'}-p_{i,j'}|) + \sum_{\tilde{q}_{j'} \ge q_{j'}} (|q_{j'}-p_{i,j'}| - |\tilde{q}_{j'}-p_{i,j'}|)\\
		&= \sum_{\tilde{q}_{j'}<q_{j'}} (\tilde{q}_{j'}-q_{j'}) + \sum_{\tilde{q}_{j'} \ge q_{j'}} (|q_{j'}-p_{i,j'}| - |\tilde{q}_{j'}-p_{i,j'}|)\\
		&\le \sum_{\tilde{q}_{j'}<q_{j'}} (\tilde{q}_{j'}-q_{j'}) + \sum_{\tilde{q}_{j'} \ge q_{j'}} (\tilde{q}_{j'}-q_{j'}) =0
		\end{align*}
	which contradicts that $\tilde{\mathbf{s}}_i$ is a better response than $\mathbf{s}_i$ for voter $i$.
	
	Next we show that the voting game has a \emph{unique} equilibrium aggregate division.\footnote{In contrast, the exact spending profile is clearly not unique. For instance, in an instance with only a single voter, that voter can enforce their belief exactly as long as they spend on each outcome in proportion to their belief, with no restriction on the magnitude of their spending.} 
	We know that the independent markets aggregate $\mathbf{q}$ is an equilibrium with spending profile $\mathbf{s}$. For contradiction, suppose that there is some other equilibrium aggregate $\tilde{\mathbf{q}}$ with spending profile $\tilde{\mathbf{s}}$. Then there is an alternative $j$ for which $\tilde{q}_j>q_j$ and an alternative $j'$ for which $\tilde{q}_{j'}<q_{j'}$. 
	Thus, there are weakly fewer voters with $p_{i,j} \ge \tilde{q}_j$ than with $p_{i,j} > q_j$. 
	From our discussion of best responses in the first part of the proof, we know that only voters with $p_{i,j} \ge \tilde{q}_j$ can have $\tilde{s}_{i,j}>0$, and that all voters with $p_{i,j} > q_j$ have $s_{i,j}=1$. It follows that $\sum_i s_{i,j} \ge \sum_i \tilde{s}_{i,j}$. 
	But, because $\tilde{q}_j>q_j$, global spending across all alternatives must be lower under $\tilde{\mathbf{s}}$ than under $\mathbf{s}$, that is, $\sum_i \sum_{\hat{j}} s_{i,\hat{j}} > \sum_i \sum_{\hat{j}} \tilde{s}_{i,\hat{j}}$. 
	By an identical argument, $\sum_i s_{i,j'} \le \sum_i \tilde{s}_{i,j'}$, implying that $\sum_i \sum_{\hat{j}} s_{i,\hat{j}} < \sum_i \sum_{\hat{j}} \tilde{s}_{i,\hat{j}}$, contradicting the previous sentence. Hence, the equilibrium aggregate division $\mathbf{q} = \smash{\mathcal{A}^{\IM}(\mathbf{P})}$ is unique.
\end{proof}

To illustrate the voting game interpretation, consider again Example~\ref{ex:im}. Define spending vectors $\mathbf{s}_1=(0,1,1), \mathbf{s}_2=(0.5,1,0), \mathbf{s}_3=(1,0,0)$. These vectors sum to $(1.5,2,1)$, which is proportional to $\mathcal{A}^{\IM}(\mathbf{P})=(\frac{1}{3},\frac{4}{9},\frac{2}{9})$. It can be checked that no voter wishes to unilaterally change their spending vector. For example, if voter 1 increases the amount she spends on alternative 1, she will increase the first coordinate of the outcome while decreasing the second and third coordinates, all of which increase the distance between the outcome and her preferred budget.

\subsection{Properties of Independent Markets}

Using the different interpretations of $\mathcal{A}^{\IM}$, we will now establish that this mechanism satisfies several appealing axioms from social choice theory.
As we mentioned at the beginning of Section~\ref{sec:im}, since $\mathcal{A}^{\IM}$ includes the uniform phantom placement, it satisfies proportionality (Definition~\ref{def:proportional}).\footnote{This condition is not necessary; a necessary and sufficient condition is that for all $i_1, i_2, \ldots, i_m$ with $\sum_{j=1}^m i_j =n$, there exists $t$ for which $f_{n-i_j}(t)=i_j/n$ for each $j \in [m]$.}
In fact, we can prove that independent markets satisfies proportionality for the more general class of \emph{$k$-approval profiles}.

\begin{definition}
	\label{def:k-approval}
	Let $k \ge 1$. 
	A budget division $\mathbf p_i$ is \emph{$k$-uniform} if it is uniform over $k$ entries.
	We say that voter $i$ approves a project $j$ if $p_{i,j} = \frac1k$.
	A mechanism satisfies \emph{$k$-proportionality} if for every profile consisting only of $k$-uniform votes, for each project $j$, the outcome spending on $j$ is proportional to the number of voters who approve $j$.
\end{definition}

For each $k\ge 1$, the phantom system of the independent markets mechanism places phantoms uniformly between $0$ and $\frac1k$ at some point. Our proof shows that any moving phantom mechanism with this property satisfies $k$-proportionality.

\begin{theorem}
	\label{thm:k-approval}
	$\mathcal{A}^{\IM}$ satisfies $k$-proportionality for all $k \le m$. In fact, $\mathcal{A}^\mathcal{F}$ is $k$-proportional, $2 \le k \le m-2$, if and only if there exists $t \in [0,1]$ with $\mathcal F(t) = (f_0(t), \frac{n-1}{n} \cdot \frac1k, \dots,  \frac{1}{n} \cdot \frac1k, 0)$, where $f_0(t) \ge \frac1k$.
\end{theorem}
\begin{proof}
	Let $k \le m$.
	Suppose that $n_j$ voters approve alternative $j$. Thus $\sum_{j=1}^m n_j = nk$, since each voter approves $k$ alternatives. When $\mathcal F(t) = (f_0(t), \frac{n-1}{n} \cdot \frac1k, \dots,  \frac{1}{n} \cdot \frac1k, 0)$ and $f_0(t) \ge \frac1k$, we have
	\[ \text{med} (\mathcal F(t), \mathbf{P}_{i \in [n], j}) = \frac{n_j}{n} \cdot \frac{1}{k}. \]
	Note that $\sum_{j=1}^m \text{med} (\mathcal F(t), \mathbf{P}_{i \in [n], j}) = 1$, so the output is normalized for phantom placement $\mathcal F(t)$. Since the output on each alternative is proportional to the number of voters approving that alternative, $\mathcal{A}^\mathcal{F}$ is $k$-proportional.
	
	Conversely, suppose there is no such $t$. Then there must exist phantoms $f_i$ and $f_h$, $i \in \{ 1, \ldots, n \}$ and $h \in \{ 0, \ldots, n-1 \}$, and $t' \in [0,1]$ for which $f_i(t') >\frac{n-i}{n} \cdot \frac1k$ and $f_h(t') < \frac{n-h}{n} \cdot \frac1k$. Assume without loss of generality that $i>h$.
	Consider a profile in which voters $1, \ldots, n-i$ approve alternatives $1, \ldots, k$, voters $n-i+1, \ldots, n-h$ approve alternatives $2, \ldots, k+1$, and voters $n-h+1, \ldots, n$ approve alternatives $3, \ldots, k+2$. This profile is well-defined since each voter approves $2 \le k \le m-2$ alternatives. 
	Further, $n-i$ voters approve alternative 1 and $n-h$ voters approve alternative 2. 
	Therefore, we have $\text{med} (\mathcal F(t'), \mathbf{P}_{i \in [n], 1}) = \min ( f_i(t'), \frac1k ) > \frac{n-i}{n} \cdot \frac1k$ and $\text{med} (\mathcal F(t'), \mathbf{P}_{i \in [n], 2}) = f_h(t') < \frac{n-h}{n} \cdot \frac1k$. Since $\text{med} (\mathcal F(t), \mathbf{P}_{i \in [n], j})$ is increasing in $t$ for every alternative $j$, there does not exist a $t$ for which $\text{med} (\mathcal F(t), \mathbf{P}_{i \in [n], 1}) =  \frac{n-i}{n} \cdot \frac1k$ and $\text{med} (\mathcal F(t), \mathbf{P}_{i \in [n], 2}) = \frac{n-h}{n} \cdot \frac1k$, as required by $k$-proportionality.
\end{proof}

An equivalent definition of $k$-proportionality says that on profiles consisting only of $k$-uniform votes, the mechanism should equal the mean of the reported distributions. This suggests a further generalization of this proportionality notion: on profiles where each voter $i$ submits a $k_i$-uniform distribution (where the $k_i$ may vary), the mechanism should equal the mean. However, it is easy to see that this property is incompatible with incentive compatibility: Let $n = m = 3$, and consider the two profiles $\mathbf P = ((1,0,0), (0.5, 0.5, 0), (0.5, 0, 0.5))$ and $\mathbf P' = ((1,0,0),(0,1,0),(0.5,0,0.5))$. If mechanism $\mathcal A$ were to satisfy the most general proportionality requirement, this would imply $\mathcal A(\mathbf P) = (\frac23, \frac16, \frac16)$ and $\mathcal A(\mathbf P') = (\frac12, \frac13, \frac16)$. Thus, $\mathcal A$ is not incentive compatible, because in profile $\mathbf P$, voter $2$ has an incentive to misreport to obtain profile $\mathbf P'$. Note that the outcome at $\mathbf P'$ is closer to voter $2$'s ideal point than the outcome at $\mathbf P$ for every project. Thus, this general proportionality criterion is incompatible even for much weaker notions of incentive compatibility.

Independent markets satisfies several other desirable properties. By using our equivalent descriptions of this mechanism, it is possible to see this easily. For example, using the voting game interpretation, it becomes clear that independent markets satisfies participation.

\begin{theorem}
	$\mathcal{A}^{\IM}$ satisfies participation.
\end{theorem}

\begin{proof}
	Under $\mathcal{A}^{\IM}$, adding a new voter who agrees with the aggregate division does not change the aggregate: $\mathcal{A}^{\IM}(\mathcal{A}^{\IM}(\mathbf{P}_{-i}),\mathbf{P}_{-i}) = \mathcal{A}^{\IM}(\mathbf{P}_{-i})$. This can be seen by noting that for the new voters, it is an equilibrium spending strategy to spend nothing, thus not changing the spending profile.
Then by incentive compatibility, $d(\mathcal{A}^{\IM}(\mathbf{p}_i,\mathbf{P}_{-i}), \mathbf{p}_i) \le d(\mathcal{A}^{\IM}(\mathcal{A}^{\IM}(\mathbf{P}_{-i}),\mathbf{P}_{-i}),\mathbf{p}_i) = d(\mathcal{A}^{\IM}(\mathbf{P}_{-i}),\mathbf{p}_i)$.
\end{proof}

Using the market interpretation, we can verify that independent markets satisfies reinforcement.

\begin{theorem}
	$\mathcal{A}^{\IM}$ satisfies reinforcement.
\end{theorem}

\begin{proof}
	Let $\mathbf{P}$ and $\mathbf{R}$ be profiles with $\mathcal{A}^{\IM}(\mathbf{P})=\mathcal{A}^{\IM}(\mathbf{R})=\mathbf{q}$. We utilize the market interpretation of independent markets. Suppose that for profile $\mathbf{P}$, market prices are normalized when $x_P^*$ goods are sold and for profile $\mathbf{R}$, market prices are normalized when $x_R^*$ goods are sold. Now consider the combined profile $\mathbf{P} \cup \mathbf{R}$ when $x^*_P+x^*_R$ goods are sold. For every alternative $j$, and every price $\pi \in [0,1]$, the total demand is equal to the total demand in profile $P$ at price $\pi$ plus the total demand in profile $R$ at price $\pi$, since a voter's demand depends only on their valuation and the price $\pi$. Likewise, the total supply is, by definition, the sum of the supplies in each individual instance. Therefore, the market clearing prices in the combined instance when $x^*_P+x^*_R$ goods are sold are equal to the (normalized) aggregate vector $\mathbf{q}$.
\end{proof}

 \section{Pareto Optimality and Social Welfare}
\label{sec:utilitarian}

We now study moving phantom mechanisms from an efficiency perspective, focusing on Pareto optimality. We will discover that there is only a single moving phantom mechanism that is Pareto optimal. This mechanism can be interpreted as maximizing utilitarian welfare. On the other hand, the independent markets mechanism is not Pareto optimal.
If voter 1 reports $(0.8, 0.2, 0)$ and voter 2 reports $(0.8, 0, 0.2)$, then independent markets chooses $(0.6, 0.2, 0.2)$, which is dominated by $(0.8, 0.1, 0.1)$. On this example, independent markets even fails to be \emph{range respecting}, which requires that $\min_{i\in [n]} \phat_{i,j} \le \mathcal{A}(\mathbf{P})_j \le \max_{i\in [n]} \phat_{i,j}$ for all $j\in[m]$. 

One can show that a moving phantom mechanism $\mathcal{A}^{\mathcal{F}}$ is range respecting if $f_0(t) = 1$ and $f_n(t) = 0$ for all $t\in[0,1]$ except for an initial period where phantom $0$ moves from $0$ to $1$ while all other phantoms remain at 0, and a period at the end where phantom $n$ moves from $0$ to $1$ while all other phantoms are at 1.\footnote{This mirrors a result in Section~\ref{sec:binary}; if the outer two phantoms are at 0 and 1, the $n-1$ remaining phantoms cannot outweigh the $n$ voter reports.}
Using this condition, it is easy to construct moving phantom mechanisms that are both proportional and range respecting. 
For instance, one can modify $\mathcal{A}^{\IM}$ so that the highest phantom is always at 1. This makes the mechanism range respecting while preserving $k$-proportionality.
Alternatively, we can take the phantom system
	\[ f_k(t) = \max\{0, 1 - (1-t)(n - k)\} \quad \text{for each $k \in \{ 0, \ldots, n \}$}, \]
that is equivalent to placing phantoms uniformly between $y$ and $1$ (rather than between $0$ and $x$, as independent markets does).\footnote{We thank an anonymous reviewer for suggesting this mechanism.}
Using Proposition~\ref{prop:uniform-overnormalizes}, it follows that $y \le 0$, and so the resulting mechanism is range respecting.
This mechanism is proportional, but it fails $k$-proportionality.

While many moving phantom mechanisms are range respecting, it is much more difficult to find a mechanism in this class which is Pareto optimal. Usually, it is possible to construct a profile in which the mechanism outputs a vector $\mathbf p$ all of whose entries are phantom reports, and then a Pareto improvement can be obtained by perturbing this vector in the directions where the majority of voter reports lie. Such a perturbation is not possible if the phantoms lie at 0 or 1, which turns out to be the only escape. As we prove below, no mechanism $\mathcal{A}^{\mathcal{F}}$ can be Pareto optimal if there is any time point $t$ when two phantoms are both strictly between 0 and 1.

This condition is extremely restrictive, and a moment's thought reveals that there is only one phantom system which avoids having two interior phantoms: All phantoms start at 0, and then, one by one, one of the phantoms is moved to 1. At each $t$, at most one phantom lies strictly between 0 and 1 while travelling. We call this phantom system $\mathcal{F}^*$. It can be formalized as
\[ f_k(t) = \begin{cases}
0 & 0 \le t \le \frac{k}{n+1},\\
t(n+1)-k & \frac{k}{n+1} <t< \frac{k+1}{n+1}, \\
1 & \frac{k+1}{n+1} \le t \le 1.
\end{cases}
\]

Below we will show that $\mathcal{A}^{\mathcal{F}^*}$ precisely corresponds to the budget aggregation mechanism that maximizes voter welfare, breaking ties in favor of the maximum entropy division.  It will immediately follow that $\mathcal{A}^{\mathcal{F}^*}$ is indeed Pareto optimal.
Combined with Theorem~\ref{thm:po} below, which shows that all other moving phantom mechanisms are Pareto inefficient, this implies that the welfare-maximizing mechanism is the unique Pareto-optimal moving phantom mechanism.

\subsection{Characterizing Pareto Optimality}

The proof of Theorem~\ref{thm:po} shows, by induction, that each phantom needs to move all the way to~1 before the next phantom can leave its position at~0. In case this does not happen, based on the approximate phantom positions, we construct a profile where the mechanism is Pareto inefficient. These constructions are of two kinds: an easier case when the interior phantoms are low (lying below $\frac{1}{n(n-1)}$), and a more involved case when one of the phantoms has moved higher. In both cases, our constructions utilize two types of alternatives. More voters report ``high'' probabilities on alternatives of the first type than on alternatives of the second type. The constructions work so that if two phantoms simultaneously take values between 0 and 1, then the mechanism outputs middling values on all alternatives. Social welfare can be improved by increasing the output on alternatives of the first type, and decreasing the output on alternatives of the second type. By incorporating enough symmetry between voters, we guarantee that social welfare gains are shared equally, so obtain a Pareto improvement.

\begin{theorem}
	\label{thm:po}
	A moving phantom mechanism $\mathcal{A}^{\mathcal{F}}$ cannot be Pareto optimal for any $m \ge n^2$ unless $\mathcal{A}^{\mathcal{F}} = \mathcal{A}^{\mathcal{F}^*}$.
\end{theorem}
\begin{proof}
	Consider a Pareto-optimal moving phantom mechanism $\mathcal{A}^\mathcal{F}$. We begin by showing that if there exists a $t$ with $f_0(t)<1$ and $f_1(t)>0$ then $\mathcal{A}^\mathcal{F}$ can be equivalently expressed as a moving phantom mechanism that does not have such a $t$. 
	It will be helpful to order the entries $\{ p_{i,j} \}$ from largest to smallest on every alternative $j$. We denote the relabeled entries $\bar{p}_{1,j} \ge \ldots \ge \bar{p}_{n,j}$.
	Suppose that $\mathcal{A}^\mathcal{F}(\mathbf{P})_j = \text{med}(f_0(t^*), \ldots, f_{n}(t^*), \phat_{1,j}, \ldots, \phat_{n,j}) <\bar{p}_{n,j}$ for some $j$. Then $\mathcal{A}^\mathcal{F}$ is not Pareto optimal, because increasing $\mathcal{A}^\mathcal{F}(\mathbf{P})_j$ and decreasing $\mathcal{A}^\mathcal{F}(\mathbf{P})_{j'}$ for any coordinate with $\mathcal{A}^\mathcal{F}(\mathbf{P})_{j'}>\bar{p}_{n,j'}$ is a Pareto improvement (such a coordinate must exist, because $\sum_j \bar{p}_{n,j} \le 1$). Therefore, for all preference profiles $\mathbf{P}$, $\mathcal{A}^\mathcal{F}(\mathbf{P})_j = \text{med}(f_0(t^*), \ldots, f_{n}(t^*), \phat_{1,j}, \ldots, \phat_{n,j}) \ge \bar{p}_{n,j}$. This implies that $f_0(t^*) \ge \bar{p}_{n,j}$ and so the exact position of $f_0(t^*)$ has no effect on the mechanism. It would be equivalent to move phantom $f_0$ to position 1 before moving phantom $f_1$.

	A very similar argument can be used to show that there cannot exist a $t$ for which $f_{n-1}(t)<1$ and $f_n(t)>0$. 
	For $n=1$ and $n=2$, this already pins down a unique Pareto-optimal moving phantom mechanism, so we assume $n \ge 3$ and focus on the intermediate phantoms for the rest of the proof.
	Suppose that there exists some index $1 \le k \le n-2$ for which $f_k(t)<1$ and $f_{k+1}(t)>0$ for some $t$. If no such $k$ exists, then $\mathcal{F} = \mathcal{F}^*$ and we are done. The remainder of the proof proceeds in two parts.

	\paragraph{Part 1.}
	We show that for any $t$ for which $f_k(t)<\frac{1}{n(n-1)}$, it must be the case that $f_{k+1}(t)=0$. 
	Suppose otherwise for contradiction, that there exists a $t'$ for which $f_k(t')<\frac{1}{n(n-1)}$ but $f_{k+1}(t')>0$.
	We define an instance with $m=n^2$ alternatives, which are divided into two kinds: we write $X = \{1, \dots, n(n-1)\}$ and $Y = \{n(n-1) + 1, ..., n^2\}$. Thus, $|X| = n(n-1)$ and $|Y| = n$. 
	For each voter $i \in [n]$, we also introduce the sets
	\begin{align*}
		X_i &= \{ ( (i-1)(n-1)+1, \ldots, (i-1)(n-1)+n^2-n-kn+k ) \text{ mod } n(n-1) \} \subseteq X, \\
		Y_i &= \{ n(n-1) + (i \text{ mod } n), n(n-1) +(i+1 \text{ mod } n), \ldots, \\
		&\qquad\qquad\qquad\qquad \dots, n(n-1)+ (i+n-k-2 \text{ mod } n) \} \subseteq Y.
	\end{align*}
	We have $|X_i|=n^2-n-kn+k=(n-1)(n-k)$ and $|Y_i|=n-k-1$.
	Voter $i \in [n]$ reports 
	\[
	\phat_{i,j} = 
	\begin{cases}
		\frac{1}{n^2-kn-1} &  j \in X_i \cup Y_i, \\
		0 & \text{otherwise.}
	\end{cases}
	\]
	Note that $\sum_{j=1}^m \phat_{i,j} = 1$ because $|X_i \cup Y_i|=n^2-kn-1$.
	Further, since each voter makes $|X_i|=(n-1)(n-k)$ non-zero reports among alternatives in  $X$, and the constructed votes treat all of these alternatives symmetrically, each alternative receives $n-k$ non-zero reports. Similarly, each alternative in $Y$ receives $n-k-1$ non-zero reports. 
	Therefore, $\text{med}(\mathcal{F}(t),\mathbf{P}_{i \in [n],j})=\min (f_k(t), \frac{1}{n^2-kn-1})$ for $j \in X$ and $\text{med}(\mathcal{F}(t),\mathbf{P}_{i \in [n],j})=\min (f_{k+1}(t), \frac{1}{n^2-kn-1})$ for $j \in Y$.
	
	By assumption, $f_k(t')<\frac{1}{n(n-1)}<\frac{1}{n^2-kn-1}$ and $f_{k+1}(t')>0$. 
	It will therefore 
	necessarily be the case that
	\begin{alignat*}{3}
		\mathcal{A}^f(\mathbf{P})_j &= f_k(t^*) < \tfrac{1}{n(n-1)} 
		&&\quad \text{for all $j \in X$, and} \\
		\mathcal{A}^f(\mathbf{P})_j &= f_{k+1}(t^*)>0
		&&\quad \text{for all $j \in Y$}.
	\end{alignat*}
	To see this, suppose that $f_k(t^*) \ge \tfrac{1}{n(n-1)}$, which implies that $t^* > t'$. In turn, this implies that $\mathcal{A}^f(\mathbf{P})_j = \min (f_{k}(t^*), \frac{1}{n^2-kn-1}) \ge \tfrac{1}{n(n-1)}$ for $j \in X$ and that $\mathcal{A}^f(\mathbf{P})_j \ge \text{med}(\mathcal{F}(t'),\mathbf{P}_{i \in [n],j}) = \min (f_{k+1}(t'), \frac{1}{n^2-kn-1}) >0$ for $j \in Y$, violating normality of $\mathcal{A}^f(\mathbf{P})$. 
	Conversely, if $f_{k+1}(t^*)=0$, then $t^*<t'$, which implies that $\mathcal{A}^f(\mathbf{P})_j = \min (f_{k}(t^*), \frac{1}{n^2-kn-1}) \le \min (f_{k}(t'), \frac{1}{n^2-kn-1}) < \tfrac{1}{n(n-1)}$ for $j \in X$, again violating normality.
	
	But this is not Pareto optimal. Consider, for some small enough $\epsilon > 0$, increasing $\mathcal{A}^\mathcal{F}(\mathbf{P})_j$ by $\epsilon$ on alternatives $j \in X$, and decreasing $\mathcal{A}^\mathcal{F}(\mathbf{P})_j$ by $\epsilon(n-1)$ on alternatives $j \in Y$. For every voter $i$, there are $|X_i|=(n-k)(n-1)$ alternatives on which the aggregate moves $\epsilon$ closer to $i$'s report, $|X|-|X_i|=k(n-1)$ alternatives for which the aggregate moves $\epsilon$ farther from $i$'s report, $|Y_i|=n-k-1$ alternatives on which the aggregate moves $\epsilon(n-1)$ farther from $i$'s report, and $|Y|-|Y_i|=k+1$ alternatives for which the aggregate moves $\epsilon(n-1)$ closer to $i$'s report. Summing these up, the change moves the aggregate 
	\[
	\left((n-k)(n-1) - k(n-1) - (n-1)(n-k-1) + (n-1)(k+1)\right)\epsilon = (2n-2)\epsilon > 0
	\]
	closer to $\mathbf{\phat}_i$ in $\ell_1$ distance. Since the new aggregate is closer to every voter, it is a Pareto improvement.

	\paragraph{Part 2.}
	We now show that if $f_k(t)<1$ then it must be the case that $f_{k+1}(t)=0$. For contradiction suppose otherwise. Let $\bar{t} = \max \{ t: f_{k+1}(t)=0 \}$ be the final snapshot at which $f_{k+1}(t)=0$. By assumption, $f_k(\bar{t})<1$. We define an instance similar to that above, defining $X$, $Y$, $X_i$, $Y_i$ as before. For each $i \in [n]$, we further partition $X_i$ into two pieces, letting
	\begin{align*}
		X_i^1&= \{ ( (i-1)(n-1)+1, \ldots, (i-1)(n-1)+n-1 ) \text{ mod } n(n-1) \} \subseteq X_i, \\
		X_i^2&=X_i \backslash X_i^1 = \{ ( (i-1)(n-1)+n, \ldots, (i-1)(n-1)+n^2-n-kn+k ) \text{ mod } n(n-1) \}.
	\end{align*}
	Note that $|X_i^1| = n-1$ and $|X_i^2| = (n-k-1)(n-1)$.
	
	Let $\delta >0$ (we will determine the exact value of $\delta$ later). We now specify the voter reports on $X$, handling the reports on alternatives in $Y$ later. For every voter $i$ and each $j \in X$, we let
	\[
		\phat_{i,j} = \begin{cases}
			\frac{1-f_k(\bar{t})-\delta}{n(n-1)-1} & j \in X_i^1 \setminus \{1\}, \\
			\frac{1-f_k(\bar{t})}{n(n-1)-1} & j \in X_i^2 \setminus \{1\}, \\
			f_k(\bar{t})+\delta & j \in (X_i^1 \cup X_i^2) \cap \{1\}, \\
			0 & j \in X \setminus (X_i^1 \cup X_i^2).
		\end{cases}
	\]
	Because we know that $f_k(\bar{t}) \ge \frac{1}{n(n-1)}$ from Part 1 of the proof, we have that 
	\begin{equation}
		\label{eq:new-value-higher}
		\frac{1-f_k(\bar{t})-\delta}{n(n-1)-1} < \frac{1-f_k(\bar{t})}{n(n-1)-1} \le \frac{1-\frac{1}{n(n-1)}}{n(n-1)-1} = \frac{1}{n(n-1)} \le f_k(\bar{t}) < f_k(\bar{t})+\delta,
	\end{equation}
	which in particular implies that we have distributed more probability mass on $X$ for voters $i$ with $1 \in X_i^1 \cup X_i^2$ than for other voters. To set $\delta$, we need to choose some value that guarantees $\sum_{j\in X} \phat_{i,j}<1$ for all $i$. By \eqref{eq:new-value-higher}, it is sufficient to set $\delta$ so that
	\begin{equation}
	\label{eq:delta-def}
	\textstyle
	f_k(\bar{t})+\delta + \left((n-k)(n-1)-1\right)\frac{1-f_k(\bar{t})}{n(n-1)-1} <1.
	\end{equation}
	To see that such a value of $\delta$ exists, note that the left-hand side of  \eqref{eq:delta-def} is continuous in $\delta$ and  takes a value strictly less than 1 when $\delta=0$:
	\begin{align*}
	\textstyle
	f_k(\bar{t}) + ((n-k)(n-1)-1)\frac{1-f_k(\bar{t})}{n(n-1)-1} &< f_k(\bar{t}) + (n(n-1)-1)\textstyle\frac{1-f_k(\bar{t})}{n(n-1)-1}\\
	&= f_k(\bar{t}) + 1 - f_k(\bar{t}) = 1.
	\end{align*}

	We now specify voter reports on the alternatives in $Y$: For each $i$, we distribute the remaining positive mass $1 - \sum_{j\in X} p_{i,j}$ evenly among $j \in Y_i$, and set $\phat_{i,j} = 0$ for $j \in Y \setminus Y_i$.
	
	As before, each alternative $j \in X$ receives $n-k$ non-zero reports. For $j=1$, all of these reports are $\phat_{i,1}=f_k(\bar{t})+\delta$. For $j \ge 2$, there exist $n-k-1$ voters with $\phat_{i,j}=\frac{1-f_k(\bar{t})}{n(n-1)-1}$ and a single voter with $\phat_{i,j}=\frac{1-f_k(\bar{t})-\delta}{n(n-1)-1}$ (since the sets $X_i^1$ are pairwise disjoint). For every alternative $j \in Y$, there are $n-k-1$ non-zero reports. It follows that
	\begin{alignat*}{3}
		\text{med}(\mathcal{F}(\bar{t}),\mathbf{P}_{i \in [n],1})&=f_k(\bar{t}), \\
		\text{med}(\mathcal{F}(\bar{t}),\mathbf{P}_{i \in [n],j})&=\tfrac{1-f_k(\bar{t})-\delta}{n(n-1)-1} &\quad& \text{for all } j \in X \setminus \{1\}, \\
		\text{med}(\mathcal{F}(\bar{t}),\mathbf{P}_{i \in [n],j})&=f_{k+1}(\bar{t})=0 &\quad& \text{for all } j \in Y.
	\end{alignat*}
	The sum of these generalized medians is $1 - \delta$. Therefore $t$ needs to increase to achieve normalization; that is, $t^*>\bar{t}$. By the definition of $\bar{t}$, for any $t > \bar{t}$, we have that $f_{k+1}(t) >0$, and therefore 
	\[
	\mathcal{A}^\mathcal{F}(\mathbf{P})_j=\text{med}(\mathcal{F}(t^*),\mathbf{P}_{i \in [n],j})>0 \quad \text{for all } j \in Y.
	\]
	From this, we can deduce that $f_k(t^*)<f_k(\bar{t})+\delta$, since otherwise we have
	\begin{alignat*}{3}
		\mathcal{A}^\mathcal{F}(\mathbf{P})_1 &=\text{med}(\mathcal{F}(t^*),\mathbf{P}_{i \in [n],1}) \ge f_k(\bar{t})+\delta, \\
		\mathcal{A}^\mathcal{F}(\mathbf{P})_j &\ge \text{med}(\mathcal{F}(\bar{t}),\mathbf{P}_{i \in [n],j})=\tfrac{1-f_k(\bar{t})-\delta}{n(n-1)-1} &\quad&\text{for all }j \in X \setminus \{1\}, \\
		\mathcal{A}^\mathcal{F}(\mathbf{P})_j&>0 &\quad&\text{for all } j \in Y,
	\end{alignat*}
	which would yield $\sum_{j=1}^m\mathcal{A}^\mathcal{F}(\mathbf{P})_j>1$. Similary, we must have $f_{k+1}(t^*)<\frac{1-f_k(\bar{t})}{n(n-1)-1}$, since otherwise
	\begin{alignat*}{3}
		\mathcal{A}^\mathcal{F}(\mathbf{P})_1 &\ge \text{med}(\mathcal{F}(\bar{t}),\mathbf{P}_{i \in [n],1})=f_k(\bar{t}) \\
		\mathcal{A}^\mathcal{F}(\mathbf{P})_j &= \text{med}(\mathcal{F}(t^*),\mathbf{P}_{i \in [n],j})=\tfrac{1-f_k(\bar{t})}{n(n-1)-1} &\quad&\text{for all }j \in X \setminus \{1\}, \\
		\mathcal{A}^\mathcal{F}(\mathbf{P})_j&>0 &\quad&\text{for all } j \in Y,
	\end{alignat*}
	again yielding $\sum_{j=1}^m\mathcal{A}^\mathcal{F}(\mathbf{P})_j>1$.
	
	To summarize, we are guaranteed that
	\begin{alignat*}{3}
		f_k(\bar{t}) &\le \mathcal{A}^\mathcal{F}(\mathbf{P})_1<f_k(\bar{t})+\delta \\
		\tfrac{1-f_k(\bar{t})-\delta}{n(n-1)-1} &\le \mathcal{A}^\mathcal{F}(\mathbf{P})_j<\tfrac{1-f_k(\bar{t})}{n(n-1)-1}
		&\quad& \text{for all $j \in X \setminus \{1\}$}, \\
		0 &<\mathcal{A}^\mathcal{F}(\mathbf{P})_j && \text{for all $j \in Y$.}
	\end{alignat*}

	Now we can define a Pareto improvement to $\mathcal{A}^\mathcal{F}(\mathbf{P})$ of the same form as previously. For some small enough $\epsilon$, increase the aggregate by $\epsilon$ on alternatives $j \in X$, and decrease the aggregate by $\epsilon(n-1)$ on alternatives $j \in Y$. 
	
	For each voter $i$ with $\phat_{i,1} = f_k(\bar{t})+\delta$, the new aggregate is $\epsilon$ better than $\mathcal{A}^\mathcal{F}(\mathbf{P})$ on alternative 1, $\epsilon$ worse on alternatives $j \in X_i^1 \backslash \{ 1 \}$, with $\phat_{i,j} = \frac{1-f_k(\bar{t})-\delta}{n(n-1)-1}$, $\epsilon$ better on alternatives $j \in X_i^2$, with $\phat_{i,j} = \frac{1-f_k(\bar{t})}{n(n-1)-1}$, and $\epsilon$ worse on alternatives $j \in X \backslash X_i$, with $\phat_{i,j} = 0$. On alternatives in $Y$, there are $k+1$ alternatives on which the aggregate moves $\epsilon(n-1)$ closer to voter $i$'s report (on alternatives for which voter $i$ reports $p_{i,j}=0$), and $n-k-1$ alternatives for which the aggregate may or may not have moved $\epsilon(n-1)$ farther from voter $i$'s report. Summing these up, we get that the aggregate has moved at least
	\[
	(1-(n-2)+(n-k-1)(n-1)-k(n-1)-(n-1)(n-k-1)+(n-1)(k+1))\epsilon = 2\epsilon > 0
	\]
	closer to $\mathbf{p}_1$ in $\ell_1$ distance.

	For each voter $i$ with $\phat_{i,1}=0$, the new aggregate is $\epsilon$ worse than $\mathcal{A}^\mathcal{F}(\mathbf{P})$ on alternatives $j \in X_i^1$, with $p_{i,j} = \frac{1-f_k(\bar{t})-\delta}{n(n-1)-1}$, and $\epsilon$ better on alternatives $j \in X_i^2$, with $p_{i,j} = \frac{1-f_k(\bar{t})}{n(n-1)-1}$, and $\epsilon$ worse on alternatives $j \in X \backslash X_i$, with $p_{i,j} = 0$. On alternatives in $Y$, there are $k+1$ alternatives on which the aggregate moves $\epsilon(n-1)$ closer to voter $i$'s report (on alternatives for which voter $i$ reports $p_{i,j}=0$), and $n-k-1$ alternatives for which the aggregate may or may not have moved $\epsilon(n-1)$ farther from voter $i$'s report. Summing these up, we get that the net change of $\ell_1$ distance from $\mathbf{p}_i$ from moving the aggregate is
	\[ (-1-(n-2)+(n-k-1)(n-1)-k(n-1)-(n-1)(n-k-1)+(n-1)(k+1))\epsilon = 0.\]
	Thus, the aggregate has moved strictly closer to some voters and weakly closer to all voters, and hence the change forms a Pareto improvement, contradicting our assumption that $\mathcal{A}^\mathcal{F}$  is Pareto optimal.

	Finally, our construction uses $m = n^2$ alternatives, but we can extend it to larger $m$ by adding dummy alternatives that no voter puts any weight on.
\end{proof}

While the proof of this result requires a large enough $m$, note that moving phantom mechanisms apply to problems with an arbitrary number of projects, so no mechanism other than $\mathcal{A}^{\mathcal{F}^*}$ can provide Pareto optimality for all cases. 

Above, we constructed a mechanism that is range respecting and proportional, though it is not Pareto optimal.
Using our characterization, it is immediate that no Pareto-optimal moving phantom mechanism can be proportional:
On profiles consisting of single-minded voters, $\mathcal{A}^{\mathcal{F}^*}$ selects a division that is also single-minded, following the plurality. Hence, it is not proportional.

\begin{corollary}
For $m\ge n^2$, no moving phantom mechanism is proportional and Pareto optimal.
\end{corollary}

\begin{figure*}[!t]
	\centering
			\begin{subfigure}[b]{0.3\textwidth}
		\begin{tikzpicture}
		[yscale=3, vote/.style={black!70,thick}]
\draw (0,0) -- (0,1); \draw (-0.05,1) -- (0.05,1); \draw (-0.05,0) -- (0.05,0); \node (y0) at (-0.2,0) {0};
		\node (y1) at (-0.2,1) {1};
		
\begin{scope}[xshift=-0.2cm]
		\fill[fill=black!10] (0.5,0) rectangle (1,1);
		\draw[black!60, thin, fill=white] (0.475,0.175) rectangle (1.025,0.225);
		\node (x1) at (0.75,-0.1) {$j_1$};
		\draw[vote] (0.5,0.3) -- (1,0.3);
		\draw[vote] (0.5,0.05) -- (1,0.05);
		\draw[vote] (0.5,0.45) -- (1,0.45);
		\draw[vote] (0.5,0.2) -- (1,0.2);
		\draw[vote] (0.5,0.1) -- (1,0.1);
\end{scope}
		
		\begin{scope}[xshift=0.6cm]
\fill[fill=black!10] (0.5,0) rectangle (1,1);
		\draw[black!60, thin, fill=white] (0.475,0.175) rectangle (1.025,0.225);
		\node (x1) at (0.75,-0.1) {$j_2$};
		\draw[vote] (0.5,0.5) -- (1,0.5);
		\draw[vote] (0.5,0.2) -- (1,0.2);
		\draw[vote] (0.5,0.05) -- (1,0.05);
		\draw[vote] (0.5,0.1) -- (1,0.1);
		\draw[vote] (0.5,0.8) -- (1,0.8);
\end{scope}
		
		\begin{scope}[xshift=1.4cm]
\fill[fill=black!10] (0.5,0) rectangle (1,1);
		\draw[black!60, thin, fill=white] (0.475,0.475) rectangle (1.025,0.525);
		\node (x1) at (0.75,-0.1) {$j_3$};
		\draw[vote] (0.5,0.2) -- (1,0.2);
		\draw[vote] (0.5,0.75) -- (1,0.75);
		\draw[vote] (0.5,0.5) -- (1,0.5);
		\draw[vote] (0.5,0.7) -- (1,0.7);
		\draw[vote] (0.5,0.1) -- (1,0.1);
\end{scope}
		
		\draw[red!85!black,very thick] (0.3,0) -- (2.4,0);
		\draw[red!85!black,very thick] (0.3,1) -- (2.4,1);
		
		\node[text width=1cm, font=\tiny, anchor=west] at (2.45,1) {$f_0, f_1, f_2$};		
		\node[text width=1cm, font=\tiny, anchor=west] at (2.45,0) {$f_3, f_4,f_5$};		
\end{tikzpicture}
\end{subfigure}
	\begin{subfigure}[b]{0.3\textwidth}
		\begin{tikzpicture}
		[yscale=3, vote/.style={black!70,thick}]
\draw (0,0) -- (0,1); \draw (-0.05,1) -- (0.05,1); \draw (-0.05,0) -- (0.05,0); \node (y0) at (-0.2,0) {0};
		\node (y1) at (-0.2,1) {1};
		
\begin{scope}[xshift=-0.2cm]
		\fill[fill=black!10] (0.5,0) rectangle (1,1);
		\draw[black!60, thin, fill=white] (0.475,0.225) rectangle (1.025,0.275);
		\node (x1) at (0.75,-0.1) {$j_1$};
		\draw[vote] (0.5,0.3) -- (1,0.3);
		\draw[vote] (0.5,0.05) -- (1,0.05);
		\draw[vote] (0.5,0.45) -- (1,0.45);
		\draw[vote] (0.5,0.2) -- (1,0.2);
		\draw[vote] (0.5,0.1) -- (1,0.1);
\end{scope}
		
		\begin{scope}[xshift=0.6cm]
\fill[fill=black!10] (0.5,0) rectangle (1,1);
		\draw[black!60, thin, fill=white] (0.475,0.225) rectangle (1.025,0.275);
		\node (x1) at (0.75,-0.1) {$j_2$};
		\draw[vote] (0.5,0.5) -- (1,0.5);
		\draw[vote] (0.5,0.2) -- (1,0.2);
		\draw[vote] (0.5,0.05) -- (1,0.05);
		\draw[vote] (0.5,0.1) -- (1,0.1);
		\draw[vote] (0.5,0.8) -- (1,0.8);
\end{scope}
		
		\begin{scope}[xshift=1.4cm]
\fill[fill=black!10] (0.5,0) rectangle (1,1);
		\draw[black!60, thin, fill=white] (0.475,0.475) rectangle (1.025,0.525);
		\node (x1) at (0.75,-0.1) {$j_3$};
		\draw[vote] (0.5,0.2) -- (1,0.2);
		\draw[vote] (0.5,0.75) -- (1,0.75);
		\draw[vote] (0.5,0.5) -- (1,0.5);
		\draw[vote] (0.5,0.7) -- (1,0.7);
		\draw[vote] (0.5,0.1) -- (1,0.1);
\end{scope}
		
		\draw[red!85!black,thick] (0.3,0.25) -- (2.4,0.25);
		\draw[red!85!black,very thick] (0.3,0.0) -- (2.4,0.0);
		\draw[red!85!black,very thick] (0.3,1) -- (2.4,1);
		
		\node[text width=1cm, font=\tiny, anchor=west] at (2.45,1) {$f_0, f_1, f_2$};		
		\node[text width=1cm, font=\tiny, anchor=west] at (2.45,0.25) {$f_3$};		
		\node[text width=1cm, font=\tiny, anchor=west] at (2.45,0) {$f_4,f_5$};		
\end{tikzpicture}
\end{subfigure}
	\begin{subfigure}[b]{0.3\textwidth}
		\begin{tikzpicture}
		[yscale=3, vote/.style={black!70,thick}]
\draw (0,0) -- (0,1); \draw (-0.05,1) -- (0.05,1); \draw (-0.05,0) -- (0.05,0); \node (y0) at (-0.2,0) {0};
		\node (y1) at (-0.2,1) {1};
		
\begin{scope}[xshift=-0.2cm]
		\fill[fill=black!10] (0.5,0) rectangle (1,1);
		\draw[black!60, thin, fill=white] (0.475,0.275) rectangle (1.025,0.325);
		\node (x1) at (0.75,-0.1) {$j_1$};
		\draw[vote] (0.5,0.3) -- (1,0.3);
		\draw[vote] (0.5,0.05) -- (1,0.05);
		\draw[vote] (0.5,0.45) -- (1,0.45);
		\draw[vote] (0.5,0.2) -- (1,0.2);
		\draw[vote] (0.5,0.1) -- (1,0.1);
\end{scope}
		
		\begin{scope}[xshift=0.6cm]
\fill[fill=black!10] (0.5,0) rectangle (1,1);
		\draw[black!60, thin, fill=white] (0.475,0.475) rectangle (1.025,0.525);
		\node (x1) at (0.75,-0.1) {$j_2$};
		\draw[vote] (0.5,0.5) -- (1,0.5);
		\draw[vote] (0.5,0.2) -- (1,0.2);
		\draw[vote] (0.5,0.05) -- (1,0.05);
		\draw[vote] (0.5,0.1) -- (1,0.1);
		\draw[vote] (0.5,0.8) -- (1,0.8);
\end{scope}
		
		\begin{scope}[xshift=1.4cm]
\fill[fill=black!10] (0.5,0) rectangle (1,1);
		\draw[black!60, thin, fill=white] (0.475,0.675) rectangle (1.025,0.725);
		\node (x1) at (0.75,-0.1) {$j_3$};
		\draw[vote] (0.5,0.2) -- (1,0.2);
		\draw[vote] (0.5,0.75) -- (1,0.75);
		\draw[vote] (0.5,0.5) -- (1,0.5);
		\draw[vote] (0.5,0.7) -- (1,0.7);
		\draw[vote] (0.5,0.1) -- (1,0.1);
\end{scope}
		
		\draw[red!85!black,very thick] (0.3,0.0) -- (2.4,0.0);
		\draw[red!85!black,very thick] (0.3,1) -- (2.4,1);
		
		\node[text width=1cm, font=\tiny, anchor=west] at (2.45,1) {$f_0, f_1, f_2, f_3$};		
		\node[text width=1cm, font=\tiny, anchor=west] at (2.45,0) {$f_4,f_5$};		
\end{tikzpicture}
\end{subfigure} \caption{Snapshots of the phantom system $\mathcal{F}^*$ with $t<t^*$ (left), $t=t^*$ (center), and $t>t^*$ (right) on an instance with $n=5$, $m=3$. 
}
\label{fig:utilitarian-visualization}
\end{figure*}
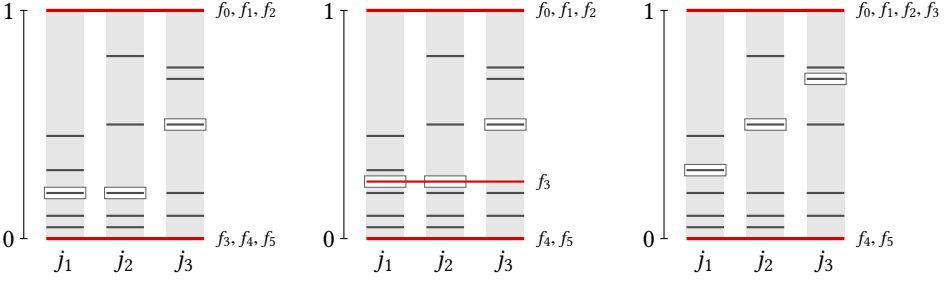

\subsection{Maximizing Utilitarian Social Welfare}

Having narrowed down the space of Pareto-optimal moving phantom mechanisms to at most one mechanism, let us examine the behavior of $\mathcal{A}^{\mathcal{F}^*}$ with the assistance of Figure~\ref{fig:utilitarian-visualization}, which takes the same form as Figure~\ref{fig:moving-phantom-example}. 
Recall from the proof of Theorem~\ref{thm:po} that $\bar{p}_{1,j} \ge \ldots \ge \bar{p}_{n,j}$ denote the entries $\{ p_{i,j} \}$, relabeled and ordered from largest to smallest.
At the snapshot of $\mathcal{F}^*$ for which $f_0(t)= \ldots = f_{k}(t)=1$ and $f_{k+1}(t)=\ldots =f_n(t)=0$, the generalized median selects the order statistic $\bar{p}_{n-k,j}$ for all $j$. We see this in Figure~\ref{fig:utilitarian-visualization} where, in the left image, $k=2$ and the generalized median is the $n-k=3\text{rd}$ highest report on each alternative, and in the right image $k=3$ and the $n-k=2\text{nd}$ highest reports are chosen.

We can think of $\mathcal{F}^*$ as partitioning the phantom ``movie'' into periods defined by which phantom is moving. Initially, all phantoms are at 0, and the generalized medians are $0$ for each $j \in [m]$. Then phantom $f_0$ moves to 1, and the generalized medians are $\bar{p}_{n,j}$. As phantom $f_k$ moves from 0 to 1, the generalized medians progress from $\bar{p}_{n-k+1,j}$ to $\bar{p}_{n-k,j}$, until all phantoms reach 1 and the generalized medians are all 1. By (a discrete analogue of) the intermediate value theorem, there must exist some value $I$ for which $\sum_{j \in [m]} \bar{p}_{I+1,j} \le 1$ and $\sum_{j \in [m]} \bar{p}_{I,j} \ge 1$, and this transition is made during the period in which phantom $n-I$ is moving. In Figure~\ref{fig:utilitarian-visualization}, we have $I=2$ because the sum of the third-highest entries is less than 1 (see the left image), while the sum of the second-highest entries is more than 1 (the right image). 

Normalization therefore occurs during the movement of phantom $f_{n-I}$, and the final value $\mathcal{A}^{\mathcal{F}^*}(\mathbf{P})_j$ lies in the interval $[\bar{p}_{I+1,j},\bar{p}_{I,j}]$. If $f_{n-I}(t^*)$ lies in this interval, then $\mathcal{A}^{\mathcal{F}^*}(\mathbf{P})_j=f_{n-I}(t^*)$, otherwise $\mathcal{A}^{\mathcal{F}^*}(\mathbf{P})_j$ is equal to the endpoint of the interval closest to $f_{n-I}(t^*)$. This is depicted in the center image of Figure~\ref{fig:utilitarian-visualization}, where $f_3(t^*)$ lies between the second and third-highest reports on the first two alternatives, but below the third-highest report on the third alternative. 

Finding the exact value of $f_{n-I}(t^*)$, and therefore the output $\mathcal{A}^{\mathcal{F}^*}(\mathbf{P})$, can be thought of as finding the ``most equal'' division, subject to interval constraints on each alternative. This problem has been studied before, and the (unique) value of $f_{n-I}(t^*)$ can be found in $O (m \log m)$ time by the Divvy algorithm of~\citet{gulati2012demand}.

Given a profile $\mathbf P$, the \emph{social cost} of an outcome $\mathbf p$ is $\sum_{i \in [n]} d(\mathbf \phat_i, \mathbf p)$, and the (utilitarian) social welfare of $\mathbf p$ is the negation of the social cost. In general, there may be multiple divisions that maximize social welfare. For example, if $m=2$, one voter reports $(1,0)$ and another reports $(0,1)$, then all divisions have the same social cost of 2. As it turns out, any division that satisfies the upper and lower bound constraints of $\bar{p}_{I,j}$ and $\bar{p}_{I+1,j}$ maximizes social welfare.
This was derived independently in a different context by~\citet{nehringpuppe}, and related results are obtained by \citet{lindner2011manipulierbarkeit}.

\begin{lemma}
	\label{lem:welfare-maximizer-characterization}
	A division $\mathbf{q}$ maximizes social welfare if and only if $\bar{p}_{I+1,j} \le q_j \le \bar{p}_{I,j}$ for all $j$.
\end{lemma}

\begin{proof}
	Let $\mathbf{q}$ be a division with $\bar{p}_{I+1,j} \le q_j = \bar{p}_{I+1,j} + \epsilon_j \le \bar{p}_{I,j}$ for all $j$, with normalization of $\mathbf{q}$ implying that $\sum_{j \in [m]} \epsilon_j = 1 - \sum_{j \in [m]} \bar{p}_{I+1,j}$. 
Then the social cost of $\mathbf{q}$ is
	\begin{align*}
		\sum_{j \in [m]} \sum_{i \in [n]} |\bar{p}_{i,j}-q_j| &= \sum_{j \in [m]} \left( \sum_{i \in [n]} |\bar{p}_{i,j}-\bar{p}_{I+1,j}| + \sum_{i \ge I+1} \epsilon_j - \sum_{i \le I} \epsilon_j \right)\\
&=\sum_{j \in [m]} \sum_{i \in [n]} |\bar{p}_{i,j}-\bar{p}_{I+1,j}| + (n-2I) \left( 1 - \sum_{j \in [m]} \bar{p}_{I+1,j}\right)
	\end{align*}
Because this expression does not depend on $\epsilon_j$, all such divisions $\mathbf{q}$ have the same social cost.

		We now show that this distance is minimal. Let $\mathbf{q}$ be a division that does not satisfy $\bar{p}_{I+1,j} \le q_j \le \bar{p}_{I,j}$ for some $j$. Suppose $q_j > \bar{p}_{I,j}$ (the case where $q_j < \bar{p}_{I+1,j}$ can be handled similarly). By the definition of $I$, there must exist some alternative $j'$ for which $q_{j'}< \bar{p}_{I,j'}$. Now, consider the division $\mathbf{q'}$ defined by $q'_j = q_j-\epsilon> \bar{p}_{I,j}$ and $q'_{j'}=q_{j'}+\epsilon< \bar{p}_{I,j'}$, with $\mathbf{q'}$ and $\mathbf{q}$ equal on all other coordinates. Compare $\mathbf{q}$ and $\mathbf{q'}$ in terms of $\ell_1$ distance from the reports. They are indistinguishable on all alternatives other than $j$ and $j'$. On alternative $j$, $\mathbf{q'}$ is $\epsilon$ closer than $\mathbf{q}$ to all entries $\bar{p}_{i,j}$ with $i \ge I$, and at most $\epsilon$ farther from all other entries. On alternative $j'$, $\mathbf{q'}$ is $\epsilon$ closer than $\mathbf{q}$ to all entries $\bar{p}_{i,j'}$ with $i \le I$, and at most $\epsilon$ farther from all other entries. Therefore, of the $2n$ entries on alternatives $j$ and $j'$, $\mathbf{q'}$ is $\epsilon$ closer than $\mathbf{q}$ to at least $n+1$ of them, and no more than $\epsilon$ farther than $\mathbf{q}$ from the other $n-1$. Therefore, $\mathbf{q}$ does not maximize social welfare.
\end{proof}

As a corollary of Lemma~\ref{lem:welfare-maximizer-characterization}, we immediately obtain that $\mathcal{A}^{\mathcal{F}^*}$ maximizes social welfare, and is therefore Pareto optimal. 
Social welfare-maximizing mechanisms have been considered before; all that is needed is a suitable tiebreaking procedure to select a single division from the set of maximizers. 
\citet{GKSA16} suggest breaking ties by selecting the lexicographically largest welfare maximizer, but this is not neutral. 
\citet{lindner2011manipulierbarkeit} proposes a different way to break ties, which is neutral: select the welfare maximizer $\mathbf p$ that is closest to the uniform distribution according to $\ell_2$ distance $\sum_{j \in [m]} (p_j-\frac{1}{m})^2$.
As we now show, it turns out that $\mathcal{A}^{\mathcal{F}^*}$ implements this tiebreaking method.

\begin{theorem}
	For every profile $\mathbf{P}$, $\mathcal{A}^{\mathcal{F}^*}$ selects the division that minimizes $\ell_2$ distance to the uniform division, among those that maximize social welfare.
\end{theorem}
\label{thm:l2-maximization}

\begin{proof}
	From the earlier discussion, we know that $\mathcal{A}^{\mathcal{F}^*}(\mathbf{P})_j = \med \{ \bar{p}_{I+1,j}, \bar{p}_{I,j}, f_{n-I}(t^*) \} \in [\bar{p}_{I+1,j}, \bar{p}_{I,j}]$ for all $j \in [m]$. Therefore, it maximizes social welfare.
	
It remains to show that, subject to these constraints, $\mathcal{A}^{\mathcal{F}^*}(\mathbf{P})$ minimizes $\ell_2$ distance to the uniform division. Consider any other division $\mathbf{q} \neq \mathcal{A}^\mathcal{\mathcal{F}^*}(\mathbf{P})$ with $\bar{p}_{I+1,j} \le q_j \le \bar{p}_{I,j}$. 
Then there must exist an alternative $j$ for which $\bar{p}_{I+1,j} \le \mathcal{A}^\mathcal{\mathcal{F}^*}(\mathbf{P})_j < q_j \le \bar{p}_{I,j}$. 
Further, because $\mathcal{A}^\mathcal{\mathcal{F}^*}(\mathbf{P})_j < \bar{p}_{I,j}$, it must be the case that $f_{n-I}(t^*) \le \mathcal{A}^\mathcal{\mathcal{F}^*}(\mathbf{P})_j= \med \{ \bar{p}_{I+1,j}, \bar{p}_{I,j}, f_{n-I}(t^*) \}$. 
There must also be an alternative $j'$ for which $\bar{p}_{I+1,j'} \le q_{j'} < \mathcal{A}^\mathcal{\mathcal{F}^*}(\mathbf{P})_{j'} \le \bar{p}_{I,j'}$, with $\mathcal{A}^\mathcal{\mathcal{F}^*}(\mathbf{P})_{j'} \le f_{n-I}(t^*)$.

Putting these together, we have that $q_{j'}<f_{n-I}(t^*)<q_j$. We also know that $\bar{p}_{I+1,j}<q_j$ and $q_{j'}<\bar{p}_{I,j'}$. Therefore, adjusting $q_j$ to $q_j - \epsilon$ and $q_{j'}$ to $q_{j'}+\epsilon$, for $\epsilon$ small enough that none of the above strict inequalities are violated, both 
(1) reduces $\ell_2$ distance to uniform,
and (2) respects social-welfare maximization. 
The former is immediately clear when $q_{j'}<\frac{1}{m}$ and $q_j>\frac{1}{m}$, and follows from strict convexity of the function $x \mapsto x^2$ when either $\frac{1}{m} \le q_{j'}$ or $q_j < \frac{1}{m}$.
Therefore, $\mathbf{q}$ is not the unique $\ell_2$ distance-minimizing division among social welfare maximizers, so $\mathcal{A}^{\mathcal{F}^*}$ is. 
\end{proof}

We note that the tiebreaking method of \citet{lindner2011manipulierbarkeit} is equivalent to choosing the welfare-maximizing division with largest Shannon entropy $- \sum_{j\in [m]} p_j \log p_j$. The proof of Theorem~\ref{thm:l2-maximization} goes through almost identically, replacing the convexity of the function $x \mapsto x^2$ by the convexity of $x \mapsto x \log x$.

Like independent markets, $\mathcal{A}^{\mathcal{F}^*}$ satisfies participation. The argument proceeds along exactly the same lines. For an individual voter, participating and reporting the existing aggregate has no effect, while reporting truthfully is at least as beneficial, by incentive compatibility.

\begin{theorem}
	$\mathcal{A}^{\mathcal{F}^*}$ satisfies participation.
\end{theorem}

The utilitarian mechanism also satisfies reinforcement, which follows from the additivity of social welfare.

\begin{theorem}
	\label{thm:util-reinforcement}
	$\mathcal{A}^{\mathcal{F}^*}$ satisfies reinforcement.
\end{theorem}

 \begin{proof}
 	Let $\mathbf{P}$ and $\mathbf{R}$ be profiles with $\mathcal{A}^\mathcal{\mathcal{F}^*}(\mathbf{P})=\mathcal{A}^\mathcal{\mathcal{F}^*}(\mathbf{R})=\mathbf{q}$. Let $\mathbf{q'}$ be some division with higher entropy than $\mathbf{q}$. Then it must be the case that $d(\mathbf{q'},\mathbf{P})>d(\mathbf{q},\mathbf{P})$ and $d(\mathbf{q'},\mathbf{R})>d(\mathbf{q},\mathbf{R})$. Since the $\ell_1$ distance is additive, $d(\mathbf{q'},\mathbf{P \cup R})>d(\mathbf{q},\mathbf{P \cup R})$.

 	Now let $\mathbf{q'}$ be some division with $d(\mathbf{q'},\mathbf{P \cup R})<d(\mathbf{q},\mathbf{P \cup R})$. Then $d(\mathbf{q'},\mathbf{P})+d(\mathbf{q'},\mathbf{R})<d(\mathbf{q},\mathbf{P})+d(\mathbf{q},\mathbf{R})$. But this is impossible, since $\mathbf{q}$ maximizes social welfare for profiles $\mathbf{P}$ and $\mathbf{R}$.

 	Therefore, any division $\mathbf{q'}\neq \mathbf{q}$ either has strictly lower social welfare on profile $\mathbf{P}\cup \mathbf{R}$, or equal social welfare but lower entropy. So $\mathcal{A}^\mathcal{\mathcal{F}^*}((\mathbf{P} \cup \mathbf{R}))=\mathbf{q}$.
 \end{proof} \section{Conclusion}
\label{sec:conclusion}

We considered the problem of aggregating budget proposals for participatory budgeting. Inspired by the generalized median mechanisms of \citet{Moul80}, we introduced the broad class of moving phantom mechanisms and proved that all mechanisms in this class are incentive compatible under $\ell_1$ voter preferences. We analyzed two moving phantom mechanisms in detail: one that maximizes social welfare while violating the natural fairness notion of proportionality, and another that satisfies proportionality while violating Pareto optimality.

In addition to the properties discussed in the main body of the paper, some other contrasts can be drawn between these two mechanisms. For example, suppose that $n$ is large, and $n/2$ voters report $(1,0)$ while the other $n/2$ report $(0,1)$. Then, under $\mathcal{A}^{\mathcal{F}^*}$, a single voter can completely dictate the outcome. Under independent markets, we can show that the ability of a single voter to affect the outcome vanishes as $n \to \infty$. 

In many participatory budgeting applications, projects will have fixed costs, and so it does not make sense to direct only a small amount of funding to such a project. One can enrich our model by allowing for funding lower bounds, so that a project receives either no funding or at least its funding lower bound. However, in this model, one can show that there does not exist an anonymous, range respecting, and incentive-compatible mechanism \citep[Section~4.7]{peters2019fair}.

It would be interesting to further analyze the utilitarian mechanism $\mathcal{A}^{\mathcal{F}^*}$. For example, can it be implemented using a voting game similar to the one we found for independent markets?
Of course, there are many other moving phantom mechanisms that we have not considered. A natural question to investigate is what other properties can be achieved by other phantom systems. And finally, as we mentioned earlier, when there are only two outcomes, we know that all (anonymous, neutral, and continuous) incentive compatible budget aggregation mechanisms can be represented as moving phantom mechanisms.  It remains an open question whether this continues to hold when the number of outcomes $m$ is more than 2.

Some other open questions are whether there exist incentive-compatible aggregation mechanisms when we define preferences with respect to the $\ell_2$ or the $\ell_\infty$ distance. Also, while generalized medians in the one-dimensional single-peaked setting are even \emph{group} strategyproof, this is not in general true for moving phantom mechanisms, and it is unclear whether sensible group strategyproof mechanisms exist for budget aggregation with $\ell_1$ preferences.
 
\section*{Acknowledgements}
We thank anonymous reviewers at JET and at the ACM EC 2019 conference for helpful suggestions that improved the presentation.
We thank Klaus Nehring and Clemens Puppe for useful discussions, and Alexandra Osipov for pointing out a mistake in an earlier version.
Compared to the working paper, Theorem~\ref{thm:rational}, Proposition~\ref{prop:uniform-overnormalizes}, and Theorem~\ref{thm:k-approval} are new.
Dominik Peters was supported by EPSRC and by ERC under grant 639945 (ACCORD).

\bibliographystyle{plainnat}
\bibliography{budget}

\end{document}